\documentclass[sigconf,nonacm]{acmart}

\makeatletter
\@ACM@balancefalse
\makeatother

\usepackage{array}
\usepackage{placeins}

\theoremstyle{remark}
\newtheorem{remark}{Remark}
\theoremstyle{definition}

\begin{document}
\raggedbottom
\vfuzz=12pt

\title{Beyond Explicit Generators: Distribution-Free Linear-Decomposition Attacks on Public-Key Encryption}

\author{Ziyan Chen}
\affiliation{%
  \institution{University of Sydney}
  \city{Sydney}
  \state{New South Wales}
  \country{Australia}}

\author{Dingxuan Zhou}
\affiliation{%
  \institution{University of Sydney}
  \city{Sydney}
  \state{New South Wales}
  \country{Australia}}

% \author{Yuqiao Wang}
% % The blank location fields preserve the current Independent Researcher
% % affiliation in the source. Add the actual city and country before camera ready.
% \affiliation[obeypunctuation=true]{%
%   \institution{Independent Researcher}
%   \city{\mbox{}}
%   \country{\mbox{}}}

\renewcommand{\shortauthors}{Chen and Zhou}

\begin{abstract}
Linear-decomposition attacks show that recovering a secret algebraic action is
often unnecessary in breaking public key scheme: once a target public state lies in a known linear span, the
same decomposition coefficients can be transferred through the unknown action
to recover the shared value.  Prior work has also observed that
random public samples may replace an explicitly constructed basis under suitable
sampling conditions.  We study a different regime motivated by public-key
encryption, where the adversary uses only the public sampling-and-evaluation
oracle available to an honest participant, the induced distribution may be
arbitrary, and the objective is to attack independently generated future
ciphertexts without recovering the entire algebraic span.

We formalize this setting as public paired samples with a fixed secret
linear transport and introduce the sampled-orbit dimension as the effective
dimension of the encryption distribution.  We give a distribution-free
one-shot recovery guarantee, derive a high-probability certificate for the
fraction of future ciphertexts covered by a fixed sampled span, and establish
the optimal sampled-span complexity
$m^\star_{\mathrm{span}}(r,\varepsilon,\delta)
=\Theta((r+\log(1/\delta))/\varepsilon)$.
We then translate these results into a generic impossibility theorem showing
that publicly samplable linear key transport with polynomial sampled-orbit
dimension is incompatible with IND--CPA security whenever the transported value
determines the decryption payload.

Finally, we instantiate the framework against the 2024 probabilistic PKE from
twisted--skew group rings.  We show that its underlying Computational
Twisted--Skew Problem admits a sampler-only linear attack using only independently
generated public protocol samples, which in turn yields plaintext recovery and
constant IND--CPA advantage.  Our experiments validate the exact linear-transport
structure, demonstrate end-to-end plaintext recovery, and illustrate that high
future-ciphertext coverage may arise well before recovery of the full algebraic
span.
\end{abstract}

% \ccsdesc[500]{Security and privacy~Cryptanalysis and other attacks}
% \ccsdesc[300]{Security and privacy~Public key encryption}


\maketitle

\section{Introduction}
\label{sec:introduction}

Public-key encryption (PKE) is a foundational cryptographic primitive: anyone
holding the public key can generate a ciphertext, while only the holder of the
secret key should be able to recover the encrypted message.  A recurring design
pattern in algebraic public-key cryptography realizes this functionality through
secret algebraic actions on public objects.  In a typical commuting-action
construction, such as the Ko--Lee paradigm~\cite{KoLeeCheonHanKangPark2000BraidPKE}, a public base point
$v$ and a secret action $A$ give a public value without showing $A$
\[
    w = Av.
\]
A fresh sender samples a secret action $B_\rho$ and forms
\[
    X_\rho = B_\rho v.
\]
When the two actions are compatible, for example when
$AB_\rho=B_\rho A$, both sides can obtain the same transported value
\[
    Y_\rho
    = B_\rho w
    = A(B_\rho v)
    = AX_\rho.
\]
Thus, from the attacker's perspective, the cryptographic computation exposes a
simple but consequential structure:
\[
    \boxed{Y_\rho = AX_\rho,}
\]
where $A$ is a fixed secret linear map and $\rho$ is fresh public randomness.
Commutativity is only one mechanism producing this identity; the results of this
paper require only the fixed key transport.

\paragraph{Linear decomposition cryptanalysis.}
A central lesson of algebraic cryptanalysis is that recovering the secret action
$A$ itself may be unnecessary.  Myasnikov and Roman'kov introduced the linear
decomposition attack (LDA), showing that shared cryptographic values can often be
recovered in polynomial time from linear relations using public
elements~\cite{MyasnikovRomankov2014LDA}.  Ben-Zvi, Kalka, and Tsaban subsequently
developed the broader algebraic-span viewpoint, obtaining polynomial-time
cryptanalysis of several noncommutative algebraic protocols~\cite{BenZviKalkaTsaban2018}.  Related work also developed generic
linear-decomposition attacks for two-sided multiplication protocols~\cite{Romankov2017TwoSided}.

The algebraic principle underlying these attacks is powerful. Suppose an adversary observes \((X_i, Y_i)_{i=1}^m\) with 
\[
    Y_i = AX_i,\qquad i=1,\ldots,m.
\]
If a target published by a sender satisfies
\[
    X^\star = \sum_{i=1}^m \alpha_i X_i,
\]
then linearity immediately transfers the same coefficients through the unknown
secret map to solve the shared key:
\[
    Y^\star
    = AX^\star
    = \sum_{i=1}^m \alpha_i Y_i.
\]
Thus, once the target lies in a known public span, the secret shared key can be recovered without solving $A$.  This coefficient-transfer
principle is a classical result.

\paragraph{From random spanning to future-ciphertext coverage.}
The use of random public samples in linear-decomposition cryptanalysis is also not new.  In particular, Roman'kov observed that randomly sampled vectors can replace an explicitly constructed basis under suitable sampling conditions, and
used this idea in the cryptanalysis of the Modified Matrix Modular
Cryptosystem~\cite{Romankov2018RandomSpanning}. 

It leaves open, however, a different question.  Honest encryption does not generally provide a uniform sampler over
an entire algebraic span.  Instead, it induces some possibly highly non-uniform
distribution $\mu$ over public states $X_\rho$.  Moreover, an attacker need not
recover the entire span in order to break a future encryption, but only the span learned from public samples to contain an independently
generated future target
\[
    X^\star\sim\mu.
\]
The relevant quantity is therefore not whether the samples generate the whole
space, but rather the \emph{future-sample failure probability}
\[
    R_\mu(W_m)
    :=
    \Pr_{X^\star\sim\mu}
    \left[
        X^\star\notin W_m
    \right],
    \qquad
    W_m:=\operatorname{span}(X_1,\ldots,X_m).
\]

The distinction can be substantial.  Consider, for example, a distribution on two
independent vectors satisfying
\[
    \Pr[X=e_1]=1-2^{-\lambda},
    \qquad
    \Pr[X=e_2]=2^{-\lambda}.
\]
Generating the entire support span
$\operatorname{span}\{e_1,e_2\}$ requires observing the exponentially rare
direction $e_2$.  In contrast, the one-dimensional span
$\operatorname{span}\{e_1\}$ already contains an independently generated future
sample with probability $1-2^{-\lambda}$.  Hence full-span recovery can be far
stronger than what cryptanalysis actually requires:
\begin{center}
\fbox{\parbox{0.96\columnwidth}{\centering
The attacker need not learn the entire algebraic span; it only needs to learn
where future ciphertexts are likely to lie.}}
\end{center}

\paragraph{A statistical view of linear cryptanalysis.}
This distinction connects algebraic cryptanalysis to distribution-free learning.
Stable sample-compression theory provides high-probability guarantees for the
population error of hypotheses reconstructed from small stable compression
sets~\cite{HannekeKontorovich2021}.  More recently, related learning-theoretic work has
explicitly considered algorithms that use sampled vectors to obtain linear
subspaces covering most of an underlying arbitrary
distribution~\cite{noivirt2026replicable}.  These results, however, are not
cryptanalytic statements: they do not consider observations
$(X,AX)$ generated by a hidden secret action, coefficient transfer to future
cryptographic targets, or formal consequences for public-key security.

This paper studies precisely this intersection:
\begin{center}
\fbox{\parbox{0.88\columnwidth}{\centering
Linear-decomposition cryptanalysis $\cap$ distribution-free learning
$\cap$ formal PKE security.}}
\end{center}
We ask whether the public sampling-and-evaluation capability already required for
honest encryption can, by itself, reveal the effect of a fixed secret linear map
on future ciphertexts.

\paragraph{Our results.}
We formalize this question through a public paired sampler induced by public states and secret actions
\[
    \operatorname{Pair}(\rho)
    =
    (X_\rho,Y_\rho),
    \qquad
    Y_\rho=AX_\rho,
\]
where $\rho$ follows an arbitrary public distribution and the adversary is given
neither $A$ nor an explicit generating set or basis for the reachable subspace.
For the span of target distribution \(\mu\),
\[
    W_\mu
    :=
    \operatorname{span}(\operatorname{supp}\mu),
    \qquad
    r_\mu:=\dim W_\mu,
\]
we call $r_\mu$ the \emph{sampled orbit dimension}.

Our first result gives a distribution-free one-shot guarantee.  For independent
$X_1,\ldots,X_m,X^\star\sim\mu$,
\[
    \Pr\!\left[
        X^\star\in
        \operatorname{span}(X_1,\ldots,X_m)
    \right]
    \ge
    1-\frac{r_\mu}{m+1}.
\]
Consequently, the corresponding transported value $Y^\star=AX^\star$ is recovered
with the same probability by classical coefficient transfer.  In particular,
$2r_\mu-1$ public samples already give success probability at least $1/2$, without
any uniformity or anti-concentration assumption on $\mu$.

We then strengthen this one-shot statement to an offline-to-future guarantee.
Viewing the rank-increasing samples as a stable compression set yields, with
probability at least $1-\delta$ over the preprocessing samples,
\[
    R_\mu(W_m)
    =
    O\!\left(
        \frac{r_\mu+\log(1/\delta)}{m}
    \right).
\]
Thus
\[
    m
    =
    O\!\left(
        \frac{r_\mu+\log(1/\delta)}{\varepsilon}
    \right)
\]
public samples suffice for a fixed learned span to recover the transported value
for at least a $1-\varepsilon$ fraction of future encryption randomness, with
confidence $1-\delta$.  We complement this upper bound with matching lower bounds
for the distribution-free sampled-span problem and obtain the minimax
characterization
\[
    m^\star_{\mathrm{span}}(r,\varepsilon,\delta)
    =
    \Theta\!\left(
        \frac{r+\log(1/\delta)}{\varepsilon}
    \right).
\]
The optimality statement concerns sampled-span coverage; we do not claim a lower
bound for all possible cryptanalysis.

Finally, we translate these statistical guarantees into a formal PKE
impossibility result.  For schemes in which a publicly samplable exact linear
transport
\[
    Y_\rho=A_{\mathsf{sk}}X_\rho
\]
determines the decryption-critical payload information, polynomial sampled-orbit
dimension implies a polynomial-time IND--CPA adversary with constant advantage.
Hence hiding an explicit algebraic generating set is insufficient: public sampling
alone can already reveal the secret transport on enough of the ciphertext
distribution to violate semantic security.

\paragraph{A concrete break of a twisted--skew PKE.}
We instantiate the framework against the probabilistic PKE proposed by
de la Cruz, Mart\'inez-Moro, Mu\~noz-Ruiz, and
Villanueva-Polanco~\cite{CruzMartinezMunozVillanueva2024}.
Their construction is based on twisted--skew group rings and uses a generally
non-associative multiplication.  Nevertheless, keeping the multiplication order
fixed, we show that the decryption-critical map is
$\mathbb F_q$-linear and that the public protocol itself provides exactly the
paired samples required by our attack.  This yields a direct sampler-only attack
on the Computational Twisted--Skew Problem (CTSP) underlying the construction and,
consequently, on the PKE.

For the proposed setting with $|G|=2n$, the ambient
$\mathbb F_q$-dimension is $4n$.  Our distribution-free bound therefore implies
that $8n-1$ independently generated public samples suffice to recover a fresh CTSP
target with probability at least $1/2$, yielding constant IND--CPA advantage for
the PKE.  To the best of our knowledge, we are unaware of a prior work that gives
an explicit cryptanalysis of the CTSP or Algorithms~6--8 of this exact 2024
construction.  We nevertheless state this claim cautiously: closely related
linear-system attacks on two-sided multiplication protocols have been studied,
including the recent work of Otero
Sanchez~\cite{OteroSanchez2025TwoSide}, and classical basis-based
linear-decomposition ideas may also be adaptable to the construction.  Our claim
is therefore not that linear decomposition itself is new, nor that sampler access
is necessary to break the scheme, but that \emph{honest public sampling alone}
suffices under an arbitrary induced distribution and admits an explicit
finite-sample security analysis.

\paragraph{Contributions.}
Our contributions can be summarized as follows.
\begin{itemize}
    \item \textbf{Sampler-only linear cryptanalysis.}
    We formulate public linear key transport through paired samples
    $(X_\rho,Y_\rho)$ with $Y_\rho=AX_\rho$ and characterize the
    distribution-free future-ciphertext coverage of the sampled span.  The resulting
    sample complexity is
    \[
        \Theta\!\left(
        \frac{r+\log(1/\delta)}{\varepsilon}
        \right)
    \]
    for the sampled-span problem.

    \item \textbf{Formal PKE consequence.}
    We show that publicly samplable exact linear key transport of polynomial sampled
    orbit dimension is incompatible with IND--CPA security whenever the transported
    value suffices to recover the encryption payload.

    \item \textbf{Concrete cryptanalysis.}
    We instantiate the framework against the 2024 twisted--skew group-ring PKE,
    directly attack its CTSP assumption using only publicly generated samples, and
    derive explicit plaintext-recovery and IND--CPA consequences.
\end{itemize}
\section{Preliminaries and Attack Setup}
\label{sec:prelim}

Figure~\ref{fig:linear-transport-attack} summarizes the complete setup used in
the paper.  Its left-hand side shows how an honest sender and receiver obtain
the same transported key through the public interface; its right-hand side
shows how an adversary replaces the secret computation by public sampling,
linear decomposition, and coefficient transfer.  We formalize these two flows
in that order.

\begin{figure*}[t]
    \centering
    \includegraphics[width=0.90\textwidth]{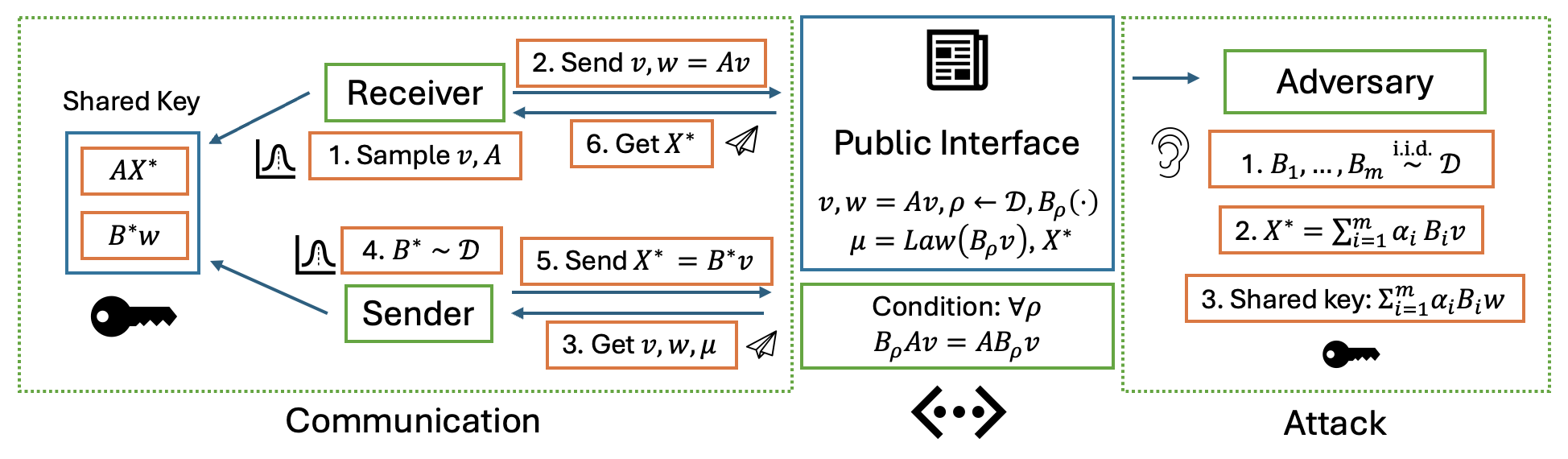}
    \caption{Honest linear key transport (left and center) and the
    sampler-only linear-decomposition attack (right).  The adversary learns the
    transported value on a covered target without recovering the secret map. Note that in the image, we use $B_\rho \sim \mathcal{D}$ as a short-hand of sampling-and-evaluation process.}
    \Description{The receiver publishes vectors v and w equal to A v. An
    honest sender samples and evaluates a public linear map B rho, sends B rho
    v, and obtains the shared key B rho w, while the receiver obtains the same
    key by applying A. The adversary samples public pairs B i v and B i w,
    decomposes a target as a linear combination of the first components, and
    transfers the coefficients to the second components.}
    \label{fig:linear-transport-attack}
\end{figure*}

\subsection{Honest public linear key transport}
\label{subsec:setup}

Let $\mathbb{F}$ be a field and let $V$ be an $N$-dimensional vector space
over $\mathbb{F}$.  All vector-space operations and linear-system computations
are assumed efficient in the natural representation of $V$.  In the
cryptographic applications below, $\mathbb{F}$ is finite.

The motivating honest protocol uses public vectors $v,w\in V$, a secret
linear map $A\in\operatorname{End}_{\mathbb{F}}(V)$, and a publicly samplable
family of linear maps, denoted as sampling-and-evaluation process,
\[
    B_\rho\in\operatorname{End}_{\mathbb{F}}(V),
    \qquad
    \rho\leftarrow\mathcal{D},
\]
where $\mathcal{D}$ is an arbitrary distribution over public sender
randomness.  The protocol shown on the left of
Figure~\ref{fig:linear-transport-attack} proceeds as follows.

\begin{enumerate}
    \item \textbf{Receiver setup.}  The receiver chooses a public base point
    $v\in V$ and a secret linear map $A$.

    \item \textbf{Public key.}  The receiver computes $w:=Av$ and publishes
    $(v,w)$ together with the public sampling-and-evaluation interface for
    $B_\rho$.

    \item \textbf{Sender initialization.}  The sender obtains
    $(v,w,\mathcal{D})$ from the public interface.

    \item \textbf{Fresh sampling.}  The sender samples
    $\rho\leftarrow\mathcal{D}$ and computes
    $X_\rho:=B_\rho v$ and sender's local shared key $Y_\rho:=B_\rho w$.

    \item \textbf{Transmission.}  The sender transmits $X_\rho$ to the
    receiver.

    \item \textbf{Receiver key derivation.}  The receiver computes reciever's shared key
    $A X_\rho$ using the secret map.
\end{enumerate}

Correctness requires the \textbf{key-transport relation} 
\begin{equation}
    B_\rho w=A(B_\rho v)
    \label{eq:transport-compatibility}
\end{equation}
for every valid $\rho$.  Consequently, the two honest computations agree:
\begin{equation*}
    \underbrace{B_\rho w}_{\text{sender key}}
    =
    \underbrace{A X_\rho}_{\text{receiver key}}.
\end{equation*}
Commutation $AB_\rho=B_\rho A$ is sufficient for
\eqref{eq:transport-compatibility}, because $w=Av$, but only the displayed
transport relation is required.

The analysis also covers protocols in which the two public values are not
syntactically obtained by applying the same map $B_\rho$ to two fixed vectors.
We therefore retain the following more general interface.  An honest sender,
and hence an adversary, can sample $\rho\leftarrow\mathcal{D}$ and run a public
procedure
\[
    \mathsf{Pair}(\rho)=(X_\rho,Y_\rho)\in V\times V.
\]
There exists a fixed secret linear map
$A\in\operatorname{End}_{\mathbb{F}}(V)$ such that
\begin{equation}
    Y_\rho=A X_\rho
    \label{eq:paired-linear-samples}
\end{equation}
for every valid $\rho$.  The protocol above realizes this abstraction
by setting
$\mathsf{Pair}(\rho)=(B_\rho v,B_\rho w)$.  The procedure need not reveal
$A$, an algebraic generating set, or a basis for the reached subspace.

Let $\mu$ denote the distribution of $X_\rho$ for
$\rho\leftarrow\mathcal{D}$.  Its \emph{sampled orbit subspace} and
\emph{sampled orbit dimension} are
\begin{equation*}
    W_\mu
    :=
    \operatorname{span}_{\mathbb{F}}
    \bigl(\operatorname{supp}(\mu)\bigr),
    \qquad
    r_\mu
    :=
    \dim_{\mathbb{F}}W_\mu
    \leq N.
\end{equation*}

\subsection{Sampler-only linear-decomposition attack}
\label{subsec:coefficient-transfer}

The right-hand side of Figure~\ref{fig:linear-transport-attack} turns the
honest public interface into an attack.  The adversary performs the following
three steps.

\begin{enumerate}
    \item \textbf{Public paired sampling.}  Sample
    $\rho_1,\ldots,\rho_m\overset{\mathrm{i.i.d.}}{\sim}\mathcal{D}$, evaluate as \(\{B_{\rho_{i}}\}_{i=1}^{m}\) and
    compute
    \[
        (X_i,Y_i):=\mathsf{Pair}(\rho_i),
        \qquad i\in[m].
    \]

    \item \textbf{Linear decomposition.}  On receiving a fresh target
    $X^\star$, test whether it lies in
    \[
        W_m:=\operatorname{span}_{\mathbb{F}}
        \{X_1,\ldots,X_m\}.
    \]
    If so, Gaussian elimination returns
    $\alpha_1,\ldots,\alpha_m\in\mathbb{F}$ such that
    \[
        X^\star=\sum_{i=1}^m\alpha_iX_i.
    \]

    \item \textbf{Key recovery.}  Transfer the same coefficients to the
    second components and output
    \[
        \widehat{Y}^\star
        :=
        \sum_{i=1}^m\alpha_iY_i.
    \]
\end{enumerate}

The algebraic justification is the classical linear-decomposition principle:
linear relations among public states transfer through the unknown fixed linear
map.  This mechanism underlies the linear decomposition attack of Myasnikov
and Roman'kov~\cite{MyasnikovRomankov2014LDA} and algebraic-span
cryptanalysis~\cite{BenZviKalkaTsaban2018}.

\begin{lemma}[Coefficient transfer]
\label{lem:coefficient-transfer}
Let $(X_i,Y_i)_{i=1}^m$ satisfy $Y_i=AX_i$ for a fixed linear map
$A\in\operatorname{End}_{\mathbb{F}}(V)$.  Suppose a target pair
$(X^\star,Y^\star)$ satisfies
\[
    Y^\star=AX^\star,
    \qquad
    X^\star=\sum_{i=1}^m\alpha_iX_i
\]
for coefficients $\alpha_1,\ldots,\alpha_m\in\mathbb{F}$.  Then the
adversary recovers the target value without recovering $A$:
\begin{equation}
    Y^\star
    =
    \sum_{i=1}^m\alpha_iY_i.
    \label{eq:coefficient-transfer}
\end{equation}
\end{lemma}

\noindent The proof is given in Appendix~\ref{app:proof-coefficient-transfer}.

Let the dimension of public sample span
\begin{equation*}
    \widehat r_m:=\dim W_m
\end{equation*}
and define the fresh-sample failure probability by
\begin{equation*}
    R_\mu(W):=\Pr_{X\sim\mu}[X\notin W].
\end{equation*}
for any subspace $W\subseteq V$. Thus $R_\mu(W_m)$ is the probability that a fresh public sample falls outside the span.

\subsection{A one-shot distribution-free attack}
\label{subsec:one-shot}

Classical linear-decomposition cryptanalysis typically constructs an exact
orbit span from publicly known generators
\cite{MyasnikovRomankov2014LDA,BenZviKalkaTsaban2018}.  Here we instead
require no assumption on the distribution $\mu$.

\begin{proposition}[One-shot attack]
\label{prop:one-shot}
Let $X_1,\ldots,X_m,X^\star\overset{\mathrm{i.i.d.}}{\sim}\mu$, with
sampled orbit dimension $r_\mu$.  Then
\begin{equation}
    \Pr\left[
        X^\star
        \in
        \operatorname{span}(X_1,\ldots,X_m)
    \right]
    \geq
    1-\frac{r_\mu}{m+1}.
    \label{eq:one-shot-bound}
\end{equation}
Consequently, using the paired samples $(X_i,Y_i)$ and
Lemma~\ref{lem:coefficient-transfer}, the target $Y^\star$ can be recovered
with probability at least
\[
    1-\frac{r_\mu}{m+1}.
\]
\end{proposition}

\noindent The proof is given in Appendix~\ref{app:proof-one-shot}.

\begin{remark}
Taking $m=2r_\mu-1$ in Proposition~\ref{prop:one-shot} gives recovery
probability at least $1/2$.  More generally, $m=O(r_\mu/\varepsilon)$ gives
one-shot failure probability at most $\varepsilon$.  We next
strengthen this pointwise statement to a high-probability certificate on the
failure rate over all future samples and establishes the optimal sample
complexity.
\end{remark}
\section{Main Results}
\label{sec:main}

We now move from one-shot recovery to a distribution-free finite-sample
analysis.  The central observation is that sequentially retaining only rank-increasing samples defines a stable sample compression scheme whose compression size is exactly the empirical sampled-orbit dimension. This connects sampler-only linear cryptanalysis to stable compression theory.

\subsection{A distribution-free future-coverage certificate}
\label{subsec:pac-certificate}

Consider the following greedy compression rule.  Process $X_1,\ldots,X_m$ in order and retain $X_i$ if and only if
\[
    X_i
    \notin
    \operatorname{span}(X_1,\ldots,X_{i-1}).
\]
Denote the retained set as
\[
    \kappa(S_m)
    \subseteq
    S_m:=\{X_1,\ldots,X_m\}.
\]
By construction, \(\kappa(S_m)\) keeps the dimension and span, namely
\begin{equation}
    |\kappa(S_m)|
    =
    \widehat r_m
    =
    \dim W_m,
    \qquad
    \operatorname{span}(\kappa(S_m))=W_m.
    \label{eq:compression-rank}
\end{equation}

\begin{theorem}[Distribution-free sampler-only certificate]
\label{thm:distribution-free}
Let $X_1,\ldots,X_m\overset{\mathrm{i.i.d.}}{\sim}\mu$. Then, for every $\delta\in(0,1)$, with probability at least $1-\delta$ over
the $m$ training samples, the fresh-sample failure probability satisfies
\begin{equation}
    R_\mu(W_m)
    \leq
    \frac{4}{m}
    \left(
        6\widehat r_m
        +
        \ln\frac{e}{\delta}
    \right).
    \label{eq:data-dependent-risk}
\end{equation}
Consequently, since $\widehat r_m\leq r_\mu$, it is further bounded by
\begin{equation}
    R_\mu(W_m)
    \leq
    \frac{4}{m}
    \left(
        6r_\mu
        +
        \ln\frac{e}{\delta}
    \right).
    \label{eq:population-rank-risk}
\end{equation}

Therefore, after observing $m$ public paired samples
$(X_i,Y_i)_{i=1}^m$, an adversary can recover the transported value $Y^\star$
for at least a $1-R_\mu(W_m)$ fraction of fresh samples with probability at
least $1-\delta$.
\end{theorem}

\noindent The proof is given in Appendix~\ref{app:proof-distribution-free}.

\subsection{Optimal sample complexity}
\label{subsec:optimal-sample}

The preceding theorem gives an $O(r_\mu/\varepsilon)$ upper bound. We next show a matching lower bound, which indicates that this dependence is unavoidable for arbitrary sampling distributions.

For integers $r\geq1$ and parameters $\varepsilon,\delta\in(0,1)$, define the mini-max distribution-free sampled-span complexity
\begin{equation*}
\begin{aligned}
    m^\star_{\mathrm{span}}(r,\varepsilon,\delta)
    := \\
    \min\biggl\{m:\,
        &\sup_{\substack{\mu:
        \dim\operatorname{span}(\operatorname{supp}\mu)\leq r}}
        \
        \Pr_{S_m\sim\mu^m}
        &
        \left[R_\mu(W_m)>\varepsilon\right]
        \leq\delta
    \biggr\}.
\end{aligned}
\end{equation*}

\begin{theorem}[Optimal distribution-free sample complexity]
\label{thm:optimal-sample-complexity}
For $r\geq2$, $\varepsilon\in(0,1/4)$, and $\delta\in(0,1/2)$,
\begin{equation*}
    m^\star_{\mathrm{span}}(r,\varepsilon,\delta)
    =
    \Theta\left(
        \frac{
            r+\ln(1/\delta)
        }{\varepsilon}
    \right).
\end{equation*}
More explicitly, from the preceding theorem,
\begin{equation}
    m^\star_{\mathrm{span}}(r,\varepsilon,\delta)
    \leq
    \left\lceil
    \frac{4}{\varepsilon}
    \left(
        6r+\ln\frac{e}{\delta}
    \right)
    \right\rceil,
    \label{eq:upper-sample-complexity}
\end{equation}
while every distribution-free guarantee must satisfy
\begin{equation}
    m^\star_{\mathrm{span}}(r,\varepsilon,\delta)
    >
    \frac{r-1}{8\varepsilon}
    \label{eq:lower-rank}
\end{equation}
and
\begin{equation}
    m^\star_{\mathrm{span}}(r,\varepsilon,\delta)
    \geq
    \frac{\ln(1/\delta)}{4\varepsilon}
    \label{eq:lower-confidence}
\end{equation}
up to integer rounding.
\end{theorem}

\noindent The proof is given in Appendix~\ref{app:proof-optimal-sample}.

\subsection{IND-CPA impossibility for public linear key transport}
\label{subsec:indcpa}

We now translate the sampled-span phenomenon into a generic cryptographic
impossibility result.

Consider a family of public-key encryption schemes indexed by the security
parameter $\lambda$.  The public key specifies a finite-dimensional vector
space $V_\lambda$, a distribution $\mathcal{D}_{pk}$ over public randomness,
and a public polynomial-time paired-sample procedure
\[
    \mathsf{Pair}(pk;\rho)=(X_\rho,Y_\rho),
    \qquad
    \rho\leftarrow\mathcal{D}_{pk}.
\]
Assume there is a fixed secret linear map
$A_{sk}\in\operatorname{End}_{\mathbb{F}}(V_\lambda)$ such that
\begin{equation}
    Y_\rho=A_{sk}X_\rho
    \label{eq:pke-transport}
\end{equation}
for every valid $\rho$.
Let the encryption randomness be $(\rho,\sigma)$, where
$\rho\leftarrow\mathcal{D}_{pk}$ determines the public linear-transport
component and $\sigma$ denotes any additional randomness used to form the
payload.  We consider ciphertexts of the general form
\begin{equation}
    \operatorname{Enc}(pk,m;\rho,\sigma)
    =
    \left(
        X_\rho,\,
        C_{\mathrm{payload}}
    \right),
    \label{eq:linear-transport-pke}
\end{equation}
and the transported value $Y_\rho$ is sufficient to recover
the plaintext from the payload.  Formally, assume there is a public
probabilistic polynomial-time algorithm $\mathsf{RecoverPayload}$ and a
negligible function $\nu$ such that, for every valid public key and every
message $m$, the probability of recovery failure is negligible, namely
\begin{equation}
    \Pr\left[
        \mathsf{RecoverPayload}
        \left(
            pk,C_{\mathrm{payload}},Y_\rho
        \right)
        =m
    \right]
    \geq
    1-\nu(\lambda),
    \label{eq:recoverable-payload}
\end{equation}
where the probability is over the encryption and recovery randomness.  We
refer to \eqref{eq:recoverable-payload} as the \emph{recoverable-payload
condition}.

The masking construction
$C_{\mathrm{payload}}=m\oplus H_\lambda(Y_\rho)$ is one special case.  The
condition also covers KEM--DEM case with encryption \(\mathsf{Enc}_{\mathrm{sym}}\) and decryption \(\mathsf{Dec}_{\mathrm{sym}}\), namely
\[
    K=\mathsf{KDF}_\lambda(Y_\rho),
    \qquad
    C_{\mathrm{payload}}
    =
    \mathsf{Enc}_{\mathrm{sym}}(K,m;\sigma),
\]
since the recovery algorithm can derive $K$ from $Y_\rho$ and run
$\mathsf{Dec}_{\mathrm{sym}}$.

\begin{theorem}[IND-CPA no-go theorem for recoverable payloads]
\label{thm:indcpa-nogo}
Suppose the scheme \eqref{eq:linear-transport-pke} satisfies the public
linear-transport condition \eqref{eq:pke-transport} and the
recoverable-payload condition \eqref{eq:recoverable-payload}.  Let
\[
    r(pk)
    :=
    \dim
    \operatorname{span}
    \left\{
        X_\rho:
        \rho\in\operatorname{supp}(\mathcal{D}_{pk})
    \right\}.
\]
Assume that there is a publicly known polynomial $R(\lambda)$ such that $r(pk)\leq R(\lambda)$
for all valid public keys, and that paired sampling and linear algebra in
$V_\lambda$ are polynomial-time operations.

Then there exists a probabilistic polynomial-time adversary with IND-CPA
advantage at least
\[
    \frac14-\nu(\lambda).
\]
Consequently, the scheme is not IND-CPA secure.  Under perfect payload
correctness, the advantage is at least $\frac14$.
\end{theorem}

\noindent The proof is given in Appendix~\ref{app:proof-indcpa-nogo}.

\begin{remark}[Stronger future-decryption interpretation]
Theorem~\ref{thm:indcpa-nogo} requires only the constant-success
one-shot bound.  Theorem~\ref{thm:distribution-free} gives a stronger
interpretation: after
\[
    m
    =
    O\left(
        \frac{
            r(pk)+\ln(1/\delta)
        }{\varepsilon}
    \right)
\]
self-generated public samples, with probability at least $1-\delta$ over the
preprocessing samples, the same sampled span permits exact recovery of the
transported value and hence payload recovery for at least a
$1-\varepsilon$ fraction of all future encryption randomness, up to the
negligible recovery error $\nu(\lambda)$.
\end{remark}
\section{Application to Twisted--Skew Group-Ring Protocols}
\label{sec:twisted-skew-pke}

We instantiate the sampler-only framework against the computational assumption
and probabilistic PKE of twisted--skew group-ring protocols proposed by
de la Cruz, Mart\'inez-Moro, Mu\~noz-Ruiz, and Villanueva-Polanco
\cite{CruzMartinezMunozVillanueva2024}.  The attack uses only public sampling
and evaluation together with linear algebra over $\mathbb{F}_q$.

\subsection{The 2024 construction}
\label{subsec:twisted-skew-construction}

Let
\[
    \mathcal{R}=\mathbb{F}_{q^2}^{\theta_\sigma,\alpha_\lambda}G
\]
be the twisted--skew group-ring space of the scheme.  We reproduce only
Algorithms~6--8 and their correctness relation.  Writing the original
two-sided action with explicit evaluation order as
\[
    \psi((a,\gamma),x):=(ax)\gamma,
\]
Alice samples
$\mathsf{sk}=(a_1,\gamma_1)$ and publishes
\begin{equation*}
    \mathsf{pk}=(a_1h)\gamma_1.
\end{equation*}
For fresh encryption randomness $\rho=(a_2,\gamma_2)$, encryption computes
\begin{equation*}
    X_\rho:=c_1=(a_2h)\gamma_2
\end{equation*}
and
\begin{equation*}
    c_2=m+Y_\rho,
    \qquad
    Y_\rho:=(a_2\mathsf{pk})\widehat{\gamma_2}.
\end{equation*}
Here and below, the displayed parentheses specify the order of multiplication;
we do not reassociate products.  Decryption subtracts
\begin{equation*}
    (a_1c_1)\widehat{\gamma_1}
\end{equation*}
from $c_2$.  The correctness theorem of the construction establishes, for
every valid encryption randomness,
\begin{equation}
    (a_1c_1)\widehat{\gamma_1}
    =
    (a_2\mathsf{pk})\widehat{\gamma_2}.
    \label{eq:ts-correctness}
\end{equation}
These are exactly the quantities computed in Algorithms~6--8 and related by
Theorem~3 of the original paper
\cite{CruzMartinezMunozVillanueva2024}.

\subsection{Fixed linear transport and self-generated pairs}
\label{subsec:twisted-skew-transport}

Define the secret-dependent, fixed-parenthesis map
\begin{equation}
    \mathcal{A}_{\mathsf{sk}}:\mathcal{R}\longrightarrow\mathcal{R},
    \qquad
    \mathcal{A}_{\mathsf{sk}}(x)
    :=
    (a_1x)\widehat{\gamma_1}.
    \label{eq:ts-secret-map}
\end{equation}

\begin{lemma}[Twisted--skew linear transport]
\label{lem:twisted-skew-transport}
The map $\mathcal{A}_{\mathsf{sk}}$ is $\mathbb{F}_q$-linear.  Moreover, every
honest encryption randomness $\rho=(a_2,\gamma_2)$ satisfies
\begin{equation*}
    Y_\rho=\mathcal{A}_{\mathsf{sk}}(X_\rho).
\end{equation*}
\end{lemma}

\noindent The proof is given in
Appendix~\ref{app:proof-twisted-skew-transport}.  Importantly, the proof uses
only $\mathbb{F}_q$-bilinearity with the parentheses in
\eqref{eq:ts-secret-map}; it never invokes associativity.

\subsection{Direct attack on the CTSP assumption}
\label{subsec:twisted-skew-ctsp}

The original paper defines the \emph{Computational Twisted--Skew Product}
(CTSP) through its Algorithm~2.  The challenger independently samples
$(a_1,\gamma_1),(a_2,\gamma_2)\leftarrow\mathsf{SK}$ and gives the adversary
\[
    pk_1=\psi((a_1,\gamma_1),h),
    \qquad
    pk_2=\psi((a_2,\gamma_2),h).
\]
The target is
\[
    k=\psi((a_2,\widehat{\gamma_2}),pk_1).
\]
For a fixed $pk_1$, write
\[
    \mathcal{A}_{pk_1}(x):=(a_1x)\widehat{\gamma_1}
\]
for the secret-dependent map associated with its hidden generating pair.
Lemma~\ref{lem:twisted-skew-transport} and the correctness identity give
\[
    k=\mathcal{A}_{pk_1}(pk_2).
\]

The adversary can independently sample $(b_i,\eta_i)\leftarrow\mathsf{SK}$
and publicly compute
\[
    X_i:=\psi((b_i,\eta_i),h),
    \qquad
    Y_i:=\psi((b_i,\widehat{\eta_i}),pk_1).
\]
The same identity gives $Y_i=\mathcal{A}_{pk_1}(X_i)$.  Let
\[
    \mu_{pk_1}:=\operatorname{Law}
    \bigl(\psi((b,\eta),h)\bigr),
    \qquad
    (b,\eta)\leftarrow\mathsf{SK},
\]
and define
\[
    r_{pk_1}
    :=
    \dim_{\mathbb{F}_q}
    \operatorname{span}_{\mathbb{F}_q}
    \bigl(\operatorname{supp}\mu_{pk_1}\bigr).
\]

\begin{corollary}[Sampler-only solution of CTSP]
\label{cor:twisted-skew-ctsp}
Given a CTSP instance $(pk_1,pk_2)$ and $m$ independently generated public
pairs $(X_i,Y_i)$ as above, a polynomial-time adversary recovers the CTSP
target with probability at least
\begin{equation}
    \Pr[\widehat{k}=k]
    \geq
    1-\frac{r_{pk_1}}{m+1}.
    \label{eq:ts-ctsp-bound}
\end{equation}
If $|G|=2n$, then $r_{pk_1}\leq4n$; hence $m=8n-1$ gives CTSP success
probability at least $1/2$.  In particular, the CTSP assumption does not hold
whenever this ambient dimension and the public operations are polynomial in
the security parameter, as in the proposed parameter families.
\end{corollary}

\noindent The proof is given in Appendix~\ref{app:proof-twisted-skew-ctsp}.
The samples use only the public sampling-and-evaluation operations available to
an honest participant; no explicit generators or precomputed basis of the
sampled action space are required.

\subsection{Plaintext recovery and IND--CPA consequence}
\label{subsec:twisted-skew-break}

The PKE challenge is precisely the CTSP computation followed by additive
masking: set $pk_1=\mathsf{pk}$, $pk_2=X^\star=c_1^\star$, and
$k=Y^\star$.  The challenge payload is $c_2^\star=m^\star+Y^\star$.
Moreover, because the message space is additive and encryption is public, an
adversary can generate the CTSP training pairs by computing
\[
    \mathsf{Enc}(\mathsf{pk},0;\rho_i)=(X_i,c_{2,i}),
    \qquad
    c_{2,i}=Y_i=\mathcal{A}_{\mathsf{sk}}(X_i).
\]
These are self-generated public encryptions.  They query neither the secret
key nor a decryption oracle.  Write $r_{\mathsf{pk}}:=r_{pk_1}$ for the
corresponding effective rank.

\begin{corollary}[Sampler-only break of the 2024 PKE]
\label{cor:twisted-skew-break}
Given $m$ independently generated encryptions of zero, a polynomial-time
adversary recovers the plaintext of an independent fresh ciphertext with
probability at least
\begin{equation}
    1-\frac{r_{\mathsf{pk}}}{m+1}.
    \label{eq:ts-recovery-bound}
\end{equation}
If $|G|=2n$, then $r_{\mathsf{pk}}\leq4n$.  Consequently, $m=8n-1$
self-generated encryptions of zero give plaintext-recovery probability at
least $1/2$ and IND--CPA advantage at least $1/4$ under the convention
\[
    \operatorname{Adv}^{\mathrm{IND\mbox{-}CPA}}
    =
    \left|\Pr[b'=b]-\frac12\right|.
\]
\end{corollary}

\noindent The proof is given in Appendix~\ref{app:proof-twisted-skew-break}.
The $8n-1$ count is a worst-case ambient-dimension guarantee: the effective
rank $r_{\mathsf{pk}}$ may be substantially smaller.  Theorem~\ref{thm:distribution-free}
also gives a high-probability certificate for the fraction of all future
ciphertexts covered by the learned span, without assuming uniformity,
coordinate independence, or anti-concentration.

\subsection{Relation to prior linear-decomposition attacks}
\label{subsec:twisted-skew-related}

Random sampling in linear-decomposition cryptanalysis is not new.
Roman'kov observed that random vectors may replace an explicitly constructed
basis under a protocol-specific uniform-sampling model
\cite{Romankov2018RandomSpanning}.  Earlier linear-decomposition and
algebraic-span attacks likewise establish the coefficient-transfer principle
\cite{MyasnikovRomankov2014LDA,BenZviKalkaTsaban2018}.  Our contribution here
is instead to use only the distribution induced by honest public sampling and
encryption, allow that distribution to be arbitrary, and quantify how a finite
sampled span covers an independently generated target.

Some generic two-sided-multiplication theorems impose hypotheses that differ
from this construction, such as associativity or explicit public generators
and bases for the relevant action spaces
\cite{Romankov2017TwoSided}.  Later attacks on earlier twisted-dihedral or
two-sided-action proposals likewise use scheme-specific hypotheses
\cite{Tinani2022TwistedDihedral,OteroSanchez2025TwoSide}.  These differences
mean that such results do not apply verbatim; they do not imply that classical
basis-based linear-decomposition methods cannot be adapted here.  We therefore
do not claim that sampler-only access is necessary to break the construction.

To the best of our knowledge, we are unaware of a prior explicit
cryptanalysis of the CTSP or Algorithms~6--8 of this exact 2024 construction.
Our contribution is a sampler-only, distribution-free attack requiring no
explicit generators or precomputed action basis, with a finite-sample CTSP
guarantee and formal plaintext-recovery and IND--CPA consequences.  We do not
claim that coefficient transfer or random-spanning ideas themselves are new.
\section{Experimental Evaluation}
\label{sec:experiments}

We organize the evaluation as one four-step validation of the cryptanalytic
story developed above:
\begin{center}
\fbox{\parbox{0.90\columnwidth}{\centering
valid setup $\longrightarrow$ end-to-end break $\longrightarrow$
explanation of the native threshold $\longrightarrow$ separation between rank
and probability-mass coverage.}}
\end{center}
The experiments use exact finite-field arithmetic; no numerical tolerance is
involved in an equality, rank, span-membership, or plaintext-recovery test.  We
implement the twisted--skew product and adjunct with the evaluation order used
by the construction, represent each ring element by its $4n$ coordinates over
$\mathbb{F}_{19}$, and use the dihedral parameter sets
$n\in\{20,23,32\}$.  For each $n$, Experiments~2--4 use 20 independently
generated public setups.  The exact sampled-orbit rank $r_{\mathsf{pk}}$ is
computed separately for every public $h$ by evaluating the bilinear sampler map
on bases of its two input spaces.

For the native-sampler curves, the sample budget ranges from
$r_{\mathsf{pk}}-10$ to $r_{\mathsf{pk}}+10$.  At each budget we generate five
independent learned spans per setup, increased to 25 spans whenever
$|m-r_{\mathsf{pk}}|\leq2$, and test 50 independent future ciphertexts against
each span.  Shaded regions and error bars are 95\% percentile intervals from a
two-stage bootstrap that resamples public setups and then learned spans within
each setup.  Thus ciphertexts tested against the same span are not incorrectly
treated as independent statistical units.

\subsection{Does the implementation realize the attack model?}
\label{subsec:exp-validation}

The first step checks that the concrete implementation has exactly the
structure required by the theory.  For the $n=20$ setup, we independently test
three identities over 200 trials: $\mathbb{F}_{19}$-linearity of the fixed
secret map, the paired-sample identity
$\mathcal{A}_{\mathsf{sk}}(X_\rho)=Y_\rho$, and encryption/decryption
correctness.  Each linearity trial samples independent $x,z\in\mathcal{R}$ and
$\alpha,\beta\in\mathbb{F}_{19}$ and checks
\[
    \mathcal{A}_{\mathsf{sk}}(\alpha x+\beta z)
    =
    \alpha\mathcal{A}_{\mathsf{sk}}(x)
    +
    \beta\mathcal{A}_{\mathsf{sk}}(z).
\]
The other tests use freshly sampled valid protocol randomness and messages.

\begin{table}[t]
    \caption{Exact algebraic validation for $p=q=19$, $n=20$, and
    $\dim_{\mathbb{F}_{19}}\mathcal{R}=80$.}
    \label{tab:algebraic-validation}
    \centering
    \small
    \begin{tabular}{lrr}
        \toprule
        Check & Trials & Violations \\
        \midrule
        $\mathbb{F}_{19}$-linearity & 200 & 0 \\
        Paired-sample transport & 200 & 0 \\
        Decryption correctness & 200 & 0 \\
        \bottomrule
    \end{tabular}
\end{table}

Table~\ref{tab:algebraic-validation} reports no violation.  This experiment is
not evidence from approximate numerical behavior: it verifies equality of the
full ring elements.  It therefore confirms that the implemented PKE exposes the
fixed linear transport assumed in Sections~\ref{subsec:setup}
and~\ref{subsec:twisted-skew-transport}.

\subsection{Does the sampler-only attack recover plaintexts?}
\label{subsec:exp-plaintext}

Having validated the setup, we next test the security consequence end to end.
For every public key, the adversary generates $m$ public encryptions of zero,
uses their pairs $(X_i,Y_i)$ to build an incremental linear-transport basis,
and attacks fresh encryptions of independently sampled messages.  When a fresh
$X^\star$ enters the learned span, the implementation transfers the exact
decomposition coefficients to obtain $Y^\star$ and outputs
$c_2^\star-Y^\star$.

\begin{figure}[t]
    \centering
    \includegraphics[width=0.86\columnwidth]{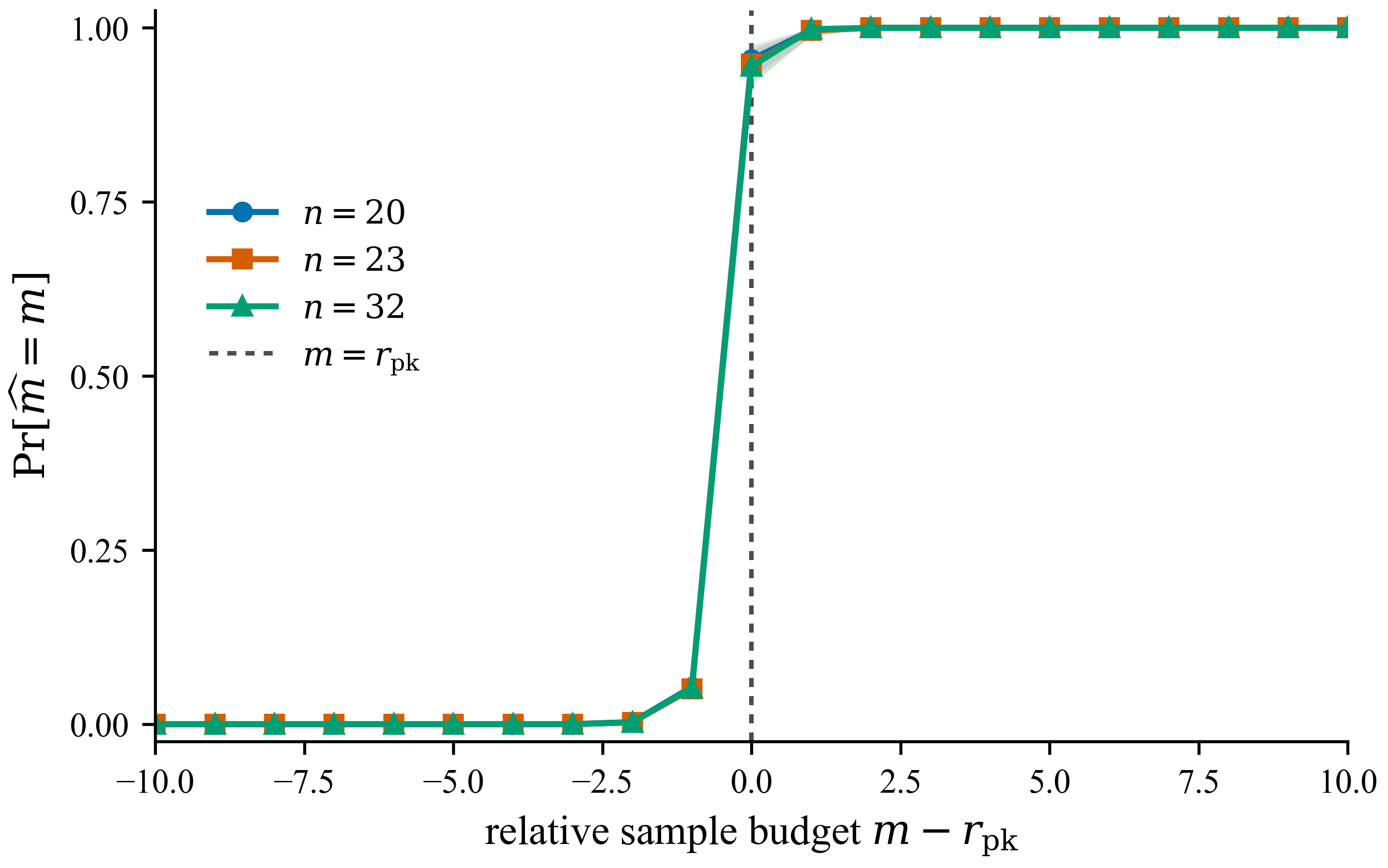}
    \caption{End-to-end plaintext recovery under the native protocol sampler.
    The horizontal coordinate is centered at the exact rank of each public
    setup.  Shading gives 95\% hierarchical-bootstrap intervals.}
    \Description{Three nearly identical recovery curves for n equal to 20, 23,
    and 32. Recovery is near zero below the exact rank, about 95 percent at the
    exact rank, and near one above it.}
    \label{fig:native-plaintext-recovery}
\end{figure}

Figure~\ref{fig:native-plaintext-recovery} shows a sharp and consistent
transition across all three parameter sizes.  At
$m=r_{\mathsf{pk}}-1$, recovery is only 5.1--5.3\%.  At
$m=r_{\mathsf{pk}}$, it rises to 95.47\%, 94.81\%, and 94.48\% for
$n=20,23,32$, respectively; one further sample raises all three rates above
99.6\%, and two further samples give 100\% in the recorded trials.  Most
importantly, every recorded span hit produced the exact challenge plaintext:
the empirical conditional probability
$\Pr[\widehat m=m\mid X^\star\in W_m]$ is one throughout the experiment.
Thus the sampled-span event is not merely a geometric diagnostic; it triggers
an actual decryption of the concrete scheme using only public encryptions.

\subsection{Why is the native recovery threshold sharp?}
\label{subsec:exp-native-threshold}

The third step explains the transition rather than repeating the attack.  We
first compute the exact sampled-orbit rank for each of the 60 public setups.
As Figure~\ref{fig:exact-rank-setups} shows, the native sampled orbit is almost
always the full ambient space: 56 of the 60 setups have
$r_{\mathsf{pk}}=4n$, while the remaining four have codimension two.  The mean
ranks are 79.8 of 80, 91.9 of 92, and 127.9 of 128 for
$n=20,23,32$.

\begin{figure}[H]
    \centering
    \includegraphics[width=0.72\columnwidth]{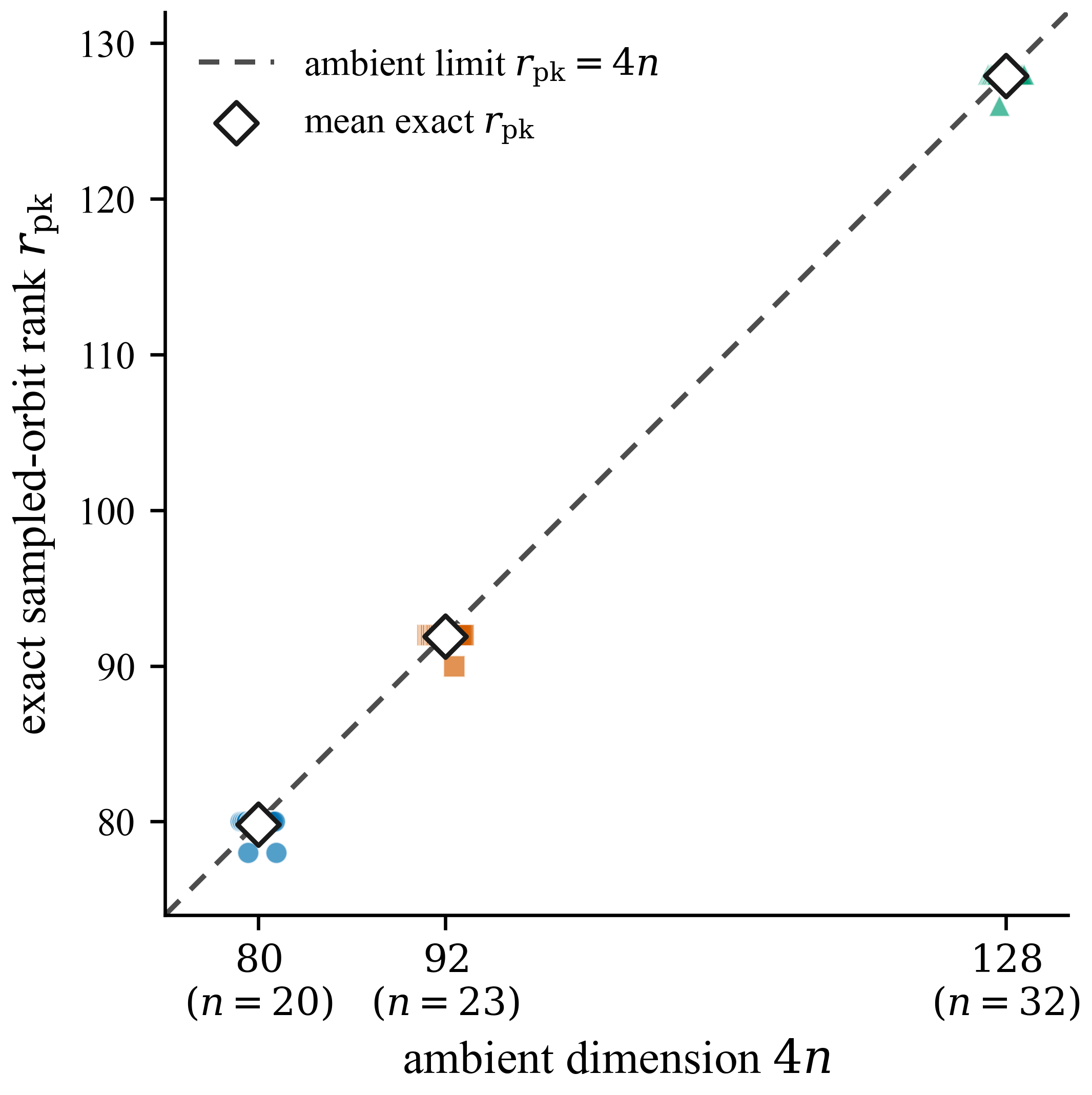}
    \caption{Exact native sampled-orbit rank over 20 independent public setups
    per parameter size.  The dashed diagonal is the ambient limit $4n$; open
    diamonds mark means.}
    \Description{Scatter plot of exact sampled-orbit rank against ambient
    dimension. Nearly all setups lie on the diagonal; four points are two rank
    units below it.}
    \label{fig:exact-rank-setups}
\end{figure}

We then measure rank acquisition and future coverage using new native samples,
independently of the plaintext-recovery experiment.  At
$m=r_{\mathsf{pk}}-1$, the learned spans already contain between 98.74\% and
99.22\% of the exact rank, yet cover only 4.94--5.37\% of future samples.  At
$m=r_{\mathsf{pk}}$, the mean rank fraction exceeds 99.93\% and future
coverage rises to 94.98--95.44\%.  Figure~\ref{fig:native-rank-fraction}, read
together with Figure~\ref{fig:native-plaintext-recovery}, therefore identifies
the mechanism:
\[
    \begin{gathered}
        \text{recovery transition}
        \\[-0.2em]
        \Longleftarrow
        \\[-0.2em]
        \text{future-coverage transition}
        \\[-0.2em]
        \Longleftarrow
        \\[-0.2em]
        \text{acquisition of the last native rank directions}.
    \end{gathered}
\]

\begin{figure}[t]
    \centering
    \includegraphics[width=\columnwidth]{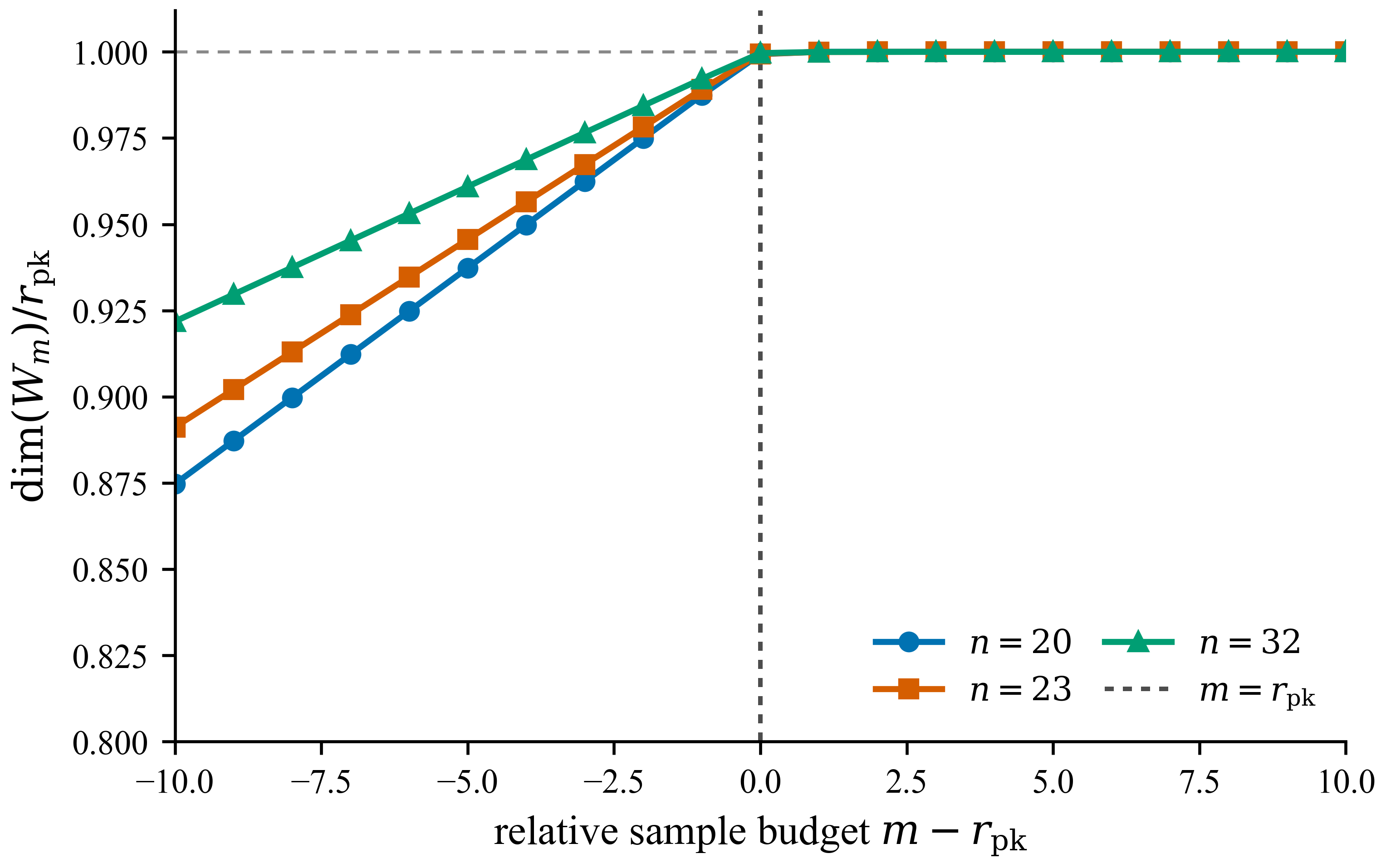}
    \caption{Fraction of the exact sampled-orbit rank learned from native
    samples.  The native recovery transition occurs only as this fraction
    approaches one.}
    \Description{Three rank-fraction curves rise gradually toward one as the
    sample budget approaches the exact rank and remain at one above it.}
    \label{fig:native-rank-fraction}
\end{figure}

This behavior is consistent with a comparatively spread-out native sampler:
missing even one of its final directions leaves substantial future probability
mass uncovered.  Full-span acquisition therefore explains the location of the
native threshold, but it does not yet establish that full rank is necessary for
an arbitrary encryption-induced distribution.

\subsection{Is full-span recovery fundamentally necessary?}
\label{subsec:exp-nonuniform}

The final step isolates that distributional question with two controlled
synthetic stress tests over valid protocol randomness.  These samplers are
\emph{not} proposed as models of the native 2024 distribution.  The moderate
sampler places 80\% of its mass uniformly on a fixed pool of 16 valid ephemeral
randomness pairs and retains the native sampler with probability 20\%.  The
strong sampler places 98\% of its mass on a pool of four valid pairs and retains
the native sampler with probability 2\%.  Because both native components are
strictly positive, the two synthetic distributions retain the full native
algebraic support and hence the same exact $r_{\mathsf{pk}}$; only the allocation
of probability mass changes.

\begin{figure*}[t]
    \centering
    \includegraphics[width=0.98\textwidth]{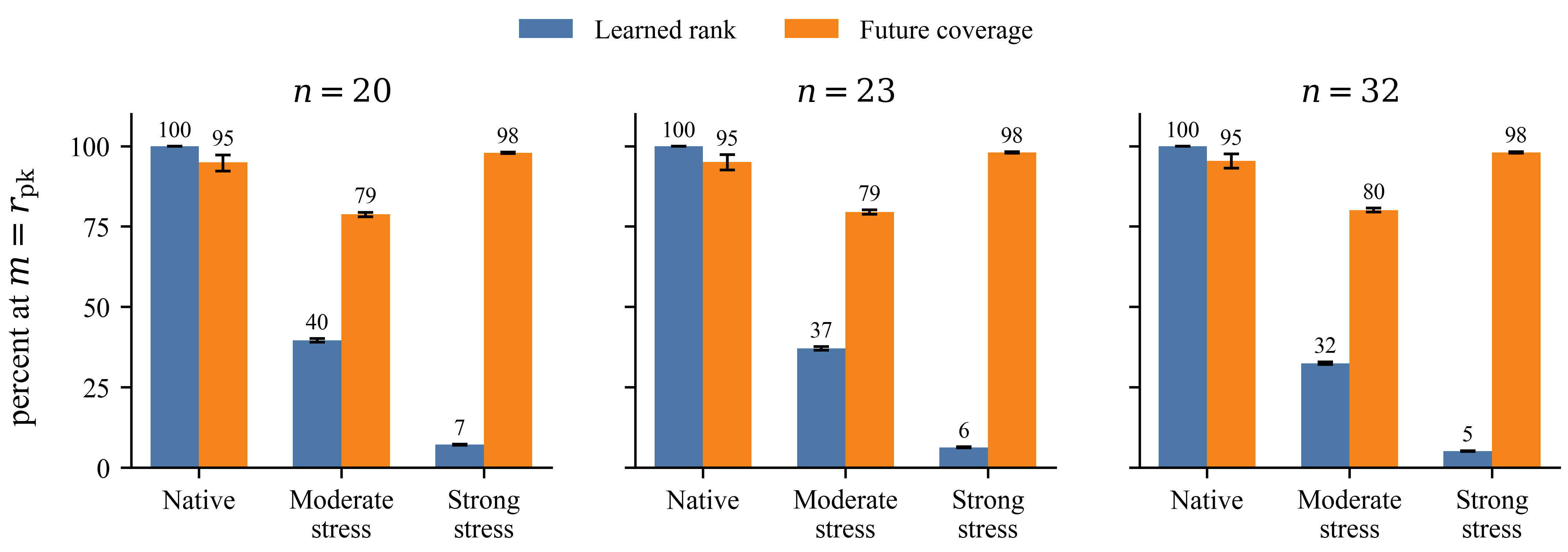}
    \caption{Learned rank fraction versus future coverage at
    $m=r_{\mathsf{pk}}$.  Native values come from Experiment~3; the moderate
    and strong hot-pool samplers are controlled synthetic stress tests.  Error
    bars give 95\% hierarchical-bootstrap intervals.}
    \Description{Grouped bars for n equal to 20, 23, and 32 compare learned
    rank and future coverage. Under strong stress, rank is only five to seven
    percent while coverage is about 98 percent.}
    \label{fig:rank-coverage-tradeoff}
\end{figure*}

Figure~\ref{fig:rank-coverage-tradeoff} displays the resulting separation.  At
$m=r_{\mathsf{pk}}$, the moderate sampler learns only 39.6\%, 37.1\%, and
32.5\% of the exact rank for $n=20,23,32$, but already covers 78.8--80.1\% of
future samples.  The strong sampler is more striking: it learns only 7.1\%,
6.3\%, and 5.1\% of the exact rank while covering 97.88--98.03\% of future
samples.  Its rank deficit grows with $n$, yet its attack-relevant coverage
remains essentially the 98\% probability mass concentrated on the hot pool.

The native experiment and the stress tests therefore answer different
questions.  The former explains why this implementation exhibits a threshold
near $m=r_{\mathsf{pk}}$; the latter demonstrates that this is a property of
the native distribution, not a general prerequisite of coefficient-transfer
attacks.  The empirical conclusion matches the quantity controlled by our
theorems:
\begin{center}
\fbox{\parbox{0.90\columnwidth}{\centering
Full-span recovery explains the native experiment, but it is not generally
required for successful future-ciphertext attacks.}}
\end{center}

\FloatBarrier
\section{Discussion}
\label{sec:discussion}

\subsection{Implications for Public-Key Encryption Design}
\label{sec:discussion-implications}

Our results suggest a simple design lesson for algebraic public-key encryption:
security should not be evaluated only in terms of whether the underlying secret
action, hidden generators, or full algebraic orbit can be recovered.  If honest
public operations expose sufficiently many examples of a fixed exact linear
transport, then an adversary may recover the effect of the secret action on
future ciphertexts without reconstructing the action itself.

The attack studied in this paper relies on three structural ingredients:
a publicly samplable source of paired observations, a fixed low-dimensional exact
linear transport across those observations, and a ciphertext payload that becomes
recoverable once the transported value is known.  Preventing the present attack
therefore requires breaking at least one of these ingredients.

\begin{table}[t]
    \caption{Design implications of the sampler-only attack.}
    \label{tab:design-implications}
    \centering
    \small
    \begin{tabular}{@{}>{\raggedright\arraybackslash}p{0.34\columnwidth}
        >{\raggedright\arraybackslash}p{0.60\columnwidth}@{}}
        \toprule
        Vulnerable ingredient & Possible design response \\
        \midrule
        Public paired sampling
        &
        Avoid exposing publicly evaluable pairs that reveal the same
        secret-dependent transport across fresh randomness. \\

        Fixed exact linear transport
        &
        Introduce genuinely nonlinear, randomness-dependent, or otherwise
        non-transferable secret transformations. \\

        Low effective dimension
        &
        Ensure that the publicly sampled states do not lie in a
        polynomial-dimensional linear representation accessible to the
        adversary. \\

        Recoverable payload from the transported value
        &
        Avoid constructions in which knowledge of the transported value
        alone enables public recovery of the plaintext or session key. \\
        \bottomrule
    \end{tabular}
\end{table}

In particular, merely hiding an explicit generating set or making full-span
reconstruction inconvenient is not sufficient.  Our analysis shows that the
relevant cryptanalytic quantity is the probability mass of future ciphertexts
covered by the sampled span.  The non-uniform stress tests further demonstrate
that this coverage may become large even when the adversary has recovered only a
small fraction of the full algebraic rank.  Hence increasing algebraic support
alone does not necessarily address the vulnerability if most encryption
probability remains concentrated on a much smaller learnable region.

\subsection{Limitations and Scope}
\label{sec:discussion-limitations}

Our results concern exact linear key transport.  The generic attack assumes that
publicly generated pairs satisfy a fixed relation of the form
$Y_\rho=A X_\rho$ over a finite-dimensional vector space.  The analysis therefore
does not directly cover schemes in which the corresponding relation is nonlinear,
depends essentially on fresh hidden randomness, or is only approximately linear.

Likewise, the optimality result in Section~\ref{sec:main} characterizes
the distribution-free sample complexity of the sampled-span coverage problem.
It does not establish a lower bound against every possible cryptanalytic
algorithm.  An adversary exploiting additional algebraic information, protocol
structure, side information, or a different representation may in principle
require fewer samples.

The concrete experiments are intended to validate the attack mechanism and its
distributional interpretation for the twisted--skew construction studied in this
paper.  The synthetic non-uniform samplers are controlled stress tests over valid
protocol randomness rather than models of the native randomness distribution.
Their purpose is to isolate the distinction between algebraic rank and
future-ciphertext probability mass, not to claim that the original 2024 scheme
uses such skewed distributions.

Finally, our concrete cryptanalysis does not imply that sampler-only access is
the only possible way to attack the twisted--skew construction.  Classical
basis-based linear-decomposition or other scheme-specific algebraic techniques
may also be adaptable.  Our contribution is instead to show that the public
sampling capability alone already suffices, under arbitrary induced
distributions, and to quantify this fact through explicit finite-sample
guarantees and formal security consequences.

Extending the analysis to noisy or approximate transport is a natural next step.
In that setting, coefficient transfer introduces accumulated error, and the
success of the attack may depend on the conditioning of the sampled
representation, the norm of the decomposition coefficients, and the error
tolerance of the payload-recovery mechanism.
\section{Related Work}
\label{sec:related-work}

\paragraph{Algebraic actions in public-key cryptography.}
Algebraic group and semigroup actions have long provided a framework for
constructing public-key protocols beyond classical number-theoretic settings.
The Ko--Lee protocol used commuting subgroup actions in braid
groups~\cite{KoLeeCheonHanKangPark2000BraidPKE}, while Anshel, Anshel, and
Goldfeld proposed a related noncommutative approach based on conjugation and
commutators~\cite{AnshelAnshelGoldfeld1999}.  Maze, Monico, and Rosenthal later
formulated Diffie--Hellman-type key exchange directly through semigroup actions
on finite sets~\cite{MazeMonicoRosenthal2007}.  These constructions motivate a
general security question that is central to our work: even when recovering the
secret algebraic action is difficult, can its effect on a fresh public state be
reconstructed from publicly available linear information?

\paragraph{Linear-algebraic and linear-decomposition cryptanalysis.}
A substantial line of work shows that the answer can be positive.
Cheon and Jun gave an early polynomial-time linear-algebraic attack on the braid
Diffie--Hellman conjugacy problem~\cite{CheonJun2003LDA}.  Shpilrain and Ushakov
emphasized more generally that solving the nominal conjugacy search problem is
neither necessary nor sufficient for attacking certain group-based
protocols~\cite{ShpilrainUshakov2006CSPUnnecessary}.  Tsaban developed the
linear-centralizer method and obtained polynomial-time cryptanalysis of several
noncommutative key-exchange problems~\cite{Tsaban2015Polynomial}.  The
linear-decomposition attack of Myasnikov and Roman'kov made the coefficient
transfer principle explicit: once a target public element is represented in a
known linear span, the same coefficients can be used to reconstruct the
corresponding secret-dependent value without recovering the secret
action~\cite{MyasnikovRomankov2014LDA}.  Roman'kov subsequently applied linear
decomposition to semidirect-product protocols~\cite{Romankov2015Semidirect} and
to generic two-sided multiplication constructions~\cite{Romankov2017TwoSided}.
Ben-Zvi, Kalka, and Tsaban broadened this perspective through cryptanalysis via
algebraic spans~\cite{BenZviKalkaTsaban2018}.

Random sampling within this paradigm is also not new.  Roman'kov showed in the
cryptanalysis of the Modified Matrix Modular Cryptosystem that random public
elements can probabilistically replace an explicitly constructed
basis~\cite{Romankov2018RandomSpanning}.  Our contribution is therefore not
random spanning or coefficient transfer themselves.  We instead consider the
distribution actually induced by public protocol sampling, make no uniformity
or anti-concentration assumption, and ask whether a finite sampled span covers an
\emph{independently generated future target}.  This changes the objective from
exact recovery of the full relevant span to distributional coverage of future
cryptographic states.

Related linear-algebraic attacks have continued to appear for modern
semidirect-product and matrix-action proposals.  For example, MAKE was rapidly
shown to admit efficient cryptanalysis~\cite{MonicoMahalanobis2020MAKE}, and
Battarbee, Kahrobaei, and Shahandashti established conditions under which
semidirect-product key exchange over matrices with noncommutative coefficient
rings remains vulnerable to linear-algebraic
methods~\cite{BattarbeeKahrobaeiShahandashti2022}.  These works further
illustrate the recurring danger of hidden algebraic operations whose effects
remain accessible through low-dimensional linear representations.

\paragraph{Twisted and skew group-ring constructions.}
Our concrete application belongs to a recent family of public-key constructions
based on twisted or skew dihedral group rings.  de la Cruz and
Villanueva-Polanco proposed key exchange, probabilistic encryption, and a KEM
over twisted dihedral group algebras~\cite{DeLaCruzVillanueva2024TwistedDihedral};
a related construction using skew dihedral group rings was subsequently proposed
by de la Cruz, Mart\'inez-Moro, and
Villanueva-Polanco~\cite{DeLaCruzMartinezVillanueva2022SkewDihedral}.
Tinani gave a polynomial-time cryptanalysis of the earlier twisted-dihedral
construction through an algebraic reduction to equations over the underlying
finite field~\cite{Tinani2022TwistedDihedral}.  More recently, Otero Sanchez
studied generic attacks on protocols based on two-sided multiplication and
analyzed several related group-ring constructions~\cite{OteroSanchez2025TwoSide}.

The 2024 construction combines twisting with skewing and may be nonassociative
for its proposed parameters~\cite{CruzMartinezMunozVillanueva2024}.
Consequently, some earlier generic theorems for two-sided multiplication do not
apply verbatim.  Classical basis-based decomposition methods may nevertheless
be adaptable, so we do not claim that sampler-only access is necessary.  To the
best of our knowledge, prior work does not give an explicit attack on the CTSP
or on Algorithms~6--8 of this exact 2024 PKE.  Our attack uses only honest
public sampling and evaluation and provides a distribution-free finite-sample
guarantee with CTSP, plaintext-recovery, and IND--CPA consequences.

\paragraph{Statistical learning of spans and subspaces.}
The statistical component of our analysis is related to learning linear
structure from samples.  Rudi, Ca\~nas, and Rosasco studied the sample complexity
of subspace learning, including PCA and spectral support estimation, under
spectral assumptions on the data distribution~\cite{RudiCanasRosasco2014}.
Stable sample compression provides a different route to distribution-free
generalization: Hanneke and Kontorovich established sharp data-dependent
generalization guarantees for stable compression
schemes~\cite{HannekeKontorovich2021}, which we use to certify future coverage of
the sampled span.  Recent work has moved even closer to our statistical object.
Noivirt, Sorrell, and Tsfadia use a procedure that outputs a subspace covering
most of an arbitrary distribution as a building block for replicable
learning~\cite{noivirt2026replicable}, while Ye, Amin, and \"Ozda\u{g}lar obtain
distribution-free guarantees for cumulative acquisition of
decision-relevant directions~\cite{YeAminOzdaglar2026}.

These learning results do not address cryptanalytic paired observations of the
form $(X,AX)$, coefficient transfer through an unknown secret map, or formal PKE
security.  Conversely, classical linear-decomposition attacks do not provide a
distribution-free finite-sample characterization of how a sampled span covers
future encryption randomness.  Our work lies at this intersection, converting
distributional span coverage into exact key-transport recovery and formal
cryptographic insecurity.
\section{Conclusion}
\label{sec:conclusion}

This work develops a distributional view of linear-decomposition cryptanalysis.
Rather than assuming that an adversary is given explicit algebraic generators or
can reconstruct the entire relevant action space, we ask what can be learned
using only the public sampling-and-evaluation interface available to an honest
participant.  The resulting viewpoint shifts the cryptanalytic objective from
recovering the full algebraic span to covering the probability mass of future
ciphertexts.

Our analysis shows that this distinction is fundamental.  A finite collection
of public samples can already provide strong guarantees for independently
generated future targets under an arbitrary encryption-induced distribution.
We characterize the optimal sample complexity of this sampled-span problem and
combine the resulting coverage guarantees with the classical
coefficient-transfer principle.  This yields a generic cryptographic consequence:
publicly samplable exact linear key transport of polynomial effective dimension
cannot provide IND--CPA security when the transported value determines the
decryption-critical payload.

We demonstrate this phenomenon on the 2024 twisted--skew group-ring PKE.  Despite
the generally non-associative multiplication of the underlying construction, its
correctness relation induces a fixed linear transport when the multiplication
order is kept unchanged.  Publicly generated protocol samples therefore suffice
to attack the Computational Twisted--Skew Problem underlying the construction,
leading directly to plaintext recovery and an IND--CPA break.  Our experiments
confirm the exact transport identity and the resulting end-to-end attack, while
also illustrating the gap between algebraic rank recovery and probability-mass
coverage under non-uniform sampling.

The broader lesson is that hiding generators or making recovery of the complete
orbit space difficult does not by itself prevent linear-decomposition
cryptanalysis.  For public-key constructions exposing exact linear key transport,
the relevant security question is not only whether the secret action can be
recovered, but whether its effect on future ciphertexts can be interpolated from
publicly generated examples.  Avoiding this attack therefore requires breaking
one of the structural ingredients enabling such interpolation, rather than merely
changing the sampling distribution or withholding an explicit algebraic basis.
Understanding whether analogous attacks remain possible under noisy, approximate,
or randomness-dependent transport is a natural direction for future work.

\newpage
\bibliographystyle{ACM-Reference-Format}
\bibliography{main}

\appendix
\section{Open Science}
\label{app:open-science}

An anonymized research artifact supporting the experimental claims in
Section~\ref{sec:experiments} accompanies the submission.  The archival
repository is available at
\url{https://anonymous.4open.science/r/Distribution-free-LDA-5DF3/}.  The artifact
contains our self-contained Python implementation of the twisted--skew
arithmetic, the concrete PKE, and the sampler-only coefficient-transfer attack;
the scripts for all four experimental stages; the raw and aggregated CSV and
JSON outputs; and the visualization script and generated figures.

The implementation performs field, ring, rank, span-membership, transport, and
plaintext-recovery computations exactly over the stated finite fields.  The
experiment scripts use the Python standard library, while figure generation
additionally uses Matplotlib.  No external dataset, network service,
proprietary software, or numerical linear-algebra tolerance is required.

For reproducibility, the artifact records the master seed and deterministically
derived setup and repetition seeds, together with the parameter sets, sample
budgets, numbers of public setups and learned spans, challenges per span,
bootstrap confidence level, and number of bootstrap repetitions.  The released
tables include per-setup and per-span results in addition to the aggregates
reported in the paper, permitting the hierarchical confidence intervals and
all four figures to be regenerated from the saved outputs.

The moderate and strong non-uniform samplers in the artifact are explicitly
marked as synthetic stress tests over valid protocol randomness; they are not
presented as the native sampler of the 2024 construction.  The artifact is
intended to reproduce the algebraic validation and empirical results, not to
replace the self-contained mathematical proofs in this paper.  No human-subject
or personal data are used.

\section{Ethical Considerations}
\label{app:ethical-considerations}

This work presents dual-use cryptanalytic techniques.  Our evaluation is
restricted to a published research construction and locally generated keys,
randomness, and ciphertexts; it does not target deployed systems, access
third-party services, or use personal data.  We provide the attack details and
research artifact to enable independent validation and to inform the design of
safer public-key encryption schemes.  The implementation is intended for
research and defensive evaluation.  Applying these techniques beyond the
construction studied here should require appropriate authorization and follow
applicable responsible-disclosure practices.

\section{Proofs for Preliminaries and Attack Setup}
\label{app:proofs-pre}

\subsection{Proof of the Coefficient-Transfer Lemma}
\label{app:proof-coefficient-transfer}

\begin{proof}
By linearity of $A$,
\[
    Y^\star
    =
    AX^\star
    =
    A\left(
        \sum_{i=1}^m \alpha_i X_i
    \right)
    =
    \sum_{i=1}^m \alpha_i AX_i.
\]
Since $Y_i=AX_i$ for every $i$,
\[
    Y^\star
    =
    \sum_{i=1}^m \alpha_i Y_i,
\]
which proves \eqref{eq:coefficient-transfer}. In the commuting-action specialization,
\[
    X_i=B_{\rho_i}v,
    \qquad
    Y_i=B_{\rho_i}w,
    \qquad
    w=Av,
\]
and $AB_{\rho_i}=B_{\rho_i}A$ implies $Y_i
    =
    B_{\rho_i}Av
    =
    AB_{\rho_i}v
    =
    AX_i.$ This is precisely the coefficient-transfer mechanism used in classical linear
decomposition attacks \cite{MyasnikovRomankov2014LDA}; the present formulation
isolates the linear relation needed later without requiring an explicit
generating set for the action family.
\end{proof}

\subsection{Proof of the One-Shot Sampled-Span Proposition}
\label{app:proof-one-shot}

\begin{proof}
For convenience, write $ X_{m+1}:=X^\star$ and define
\[
    S_j
    :=
    \operatorname{span}(X_1,\ldots,X_j),
    \qquad
    S_0:=\{0\}.
\]
Let
\[
    I_j
    :=
    \mathbf{1}\{X_j\notin S_{j-1}\},
    \qquad
    q_j:=\mathbb{E}[I_j].
\]
Every event $I_j=1$ increases the dimension of the running span by exactly
one.  Hence
\[
    \sum_{j=1}^{m+1} I_j
    =
    \dim S_{m+1}
    \leq r_\mu,
\]
and therefore
\begin{equation}
    \sum_{j=1}^{m+1} q_j
    \leq r_\mu.
    \label{eq:sum-rank-increments}
\end{equation}

We next observe that $(q_j)$ is non-increasing.  Conditional on $S_{j-1}$, independence of $X_j$ gives
\[
    \Pr[X_j\notin S_{j-1}\mid S_{j-1}]
    =
    1-\mu(S_{j-1}),
\]
therefore
\[
    q_j
    =
    \mathbb{E}\bigl[1-\mu(S_{j-1})\bigr].
\]
Since
\[
    S_{j-1}\subseteq S_j,
\]
we have
\[
    1-\mu(S_j)
    \leq
    1-\mu(S_{j-1})
\]
pointwise, and thus
\[
    q_{j+1}\leq q_j.
\]
It follows from \eqref{eq:sum-rank-increments} that
\[
    q_{m+1}
    \leq
    \frac{1}{m+1}
    \sum_{j=1}^{m+1}q_j
    \leq
    \frac{r_\mu}{m+1}.
\]
Finally,
\[
    q_{m+1}
    =
    \Pr[
        X^\star
        \notin
        \operatorname{span}(X_1,\ldots,X_m)
    ],
\]
which proves \eqref{eq:one-shot-bound}.

Whenever $X^\star\in W_m$, Gaussian elimination yields coefficients $\alpha_i$ such that
\[
    X^\star=\sum_i\alpha_iX_i.
\]
Lemma~\ref{lem:coefficient-transfer} then recovers
\[
    Y^\star=\sum_i\alpha_iY_i.
\]
\end{proof}

\section{Proofs for the Main Results}
\label{app:proofs-main}

\subsection{Proof of the Distribution-Free Sampler-Only Certificate}
\label{app:proof-distribution-free}

\begin{proof}
We reduce the sampled-span procedure to a realizable stable sample compression
scheme and then apply Corollary~11 of Hanneke and Kontorovich
\cite{HannekeKontorovich2021}.

Take the instance space to be $V$ and assign every training vector the same
label $0$.  For any finite set $K\subseteq V$, define the reconstructed classifier
\[
    h_K(x)
    :=
    \mathbf{1}
    \left\{
        x\notin\operatorname{span}(K)
    \right\}.
\]
The compression map $\kappa$ retains a basis of the sample span as described
above, and the reconstruction map returns $h_{\kappa(S_m)}$.

First, the scheme is sample-consistent.  Indeed,
\[
    X_i\in
    \operatorname{span}(\kappa(S_m))
    =
    W_m
\]
for every $i$, so
\[
    h_{\kappa(S_m)}(X_i)=0.
\]
Thus the empirical classification error is zero.

Second, the compression scheme is stable.  Suppose
\[
    \kappa(S_m)
    \subseteq
    S'
    \subseteq
    S_m.
\]
Since $\kappa(S_m)$ already spans $S_m$,
\[
    \operatorname{span}(S')
    =
    \operatorname{span}(S_m).
\]
The compression rule applied to $S'$ returns some basis of
$\operatorname{span}(S')$, and hence
\[
    \operatorname{span}(\kappa(S'))
    =
    \operatorname{span}(\kappa(S_m)).
\]
Therefore the reconstructed classifiers agree:
\[
    h_{\kappa(S')}
    =
    h_{\kappa(S_m)}.
\]
This is precisely the stability property required by stable sample
compression.

Finally, by \eqref{eq:compression-rank}, $|\kappa(S_m)|=\widehat r_m$. Corollary~11 of Hanneke and Kontorovich
\cite{HannekeKontorovich2021} states that for any realizable stable
compression scheme, with probability at least $1-\delta$,
\[
    R
    \leq
    \frac{4}{m}
    \left(
        6|\kappa(S_m)|
        +
        \ln\frac{e}{\delta}
    \right).
\]
For our classifier,
\[
    R
    =
    \Pr_{X\sim\mu}
    [
        X\notin W_m
    ]
    =
    R_\mu(W_m).
\]
Substituting $|\kappa(S_m)|=\widehat r_m$ yields \eqref{eq:data-dependent-risk}.  The inequality
$\widehat r_m\leq r_\mu$ immediately gives \eqref{eq:population-rank-risk}.

For any fresh $X^\star\in W_m$, Lemma~\ref{lem:coefficient-transfer} recovers the corresponding $Y^\star$, completing the cryptanalytic interpretation.
\end{proof}

\subsection{Proof of the Optimal Sample-Complexity Theorem}
\label{app:proof-optimal-sample}

\begin{proof}
\noindent\emph{Upper bound.}\par
By Theorem~\ref{thm:distribution-free}, with probability at least $1-\delta$,
\[
    R_\mu(W_m)
    \leq
    \frac{4}{m}
    \left(
        6r+\ln\frac{e}{\delta}
    \right)
\]
for every distribution whose sampled orbit dimension is at most $r$. Thus \eqref{eq:upper-sample-complexity} is sufficient to ensure
\[
    R_\mu(W_m)\leq\varepsilon.
\]

\noindent\emph{Dimension-dependent lower bound.}\par
The argument follows the classical rare-types construction used in
sample-compression lower bounds; a recent structurally analogous use appears
in~\cite{YeAminOzdaglar2026}.

Work in $V=\mathbb{F}^r$ with linearly independent vectors $ e_1,\ldots,e_r$. Let $k=r-1$ and define
\[
    \Pr[X=e_1]
    =
    1-2\varepsilon,
\]
and
\[
    \Pr[X=e_i]
    =
    \frac{2\varepsilon}{k},
    \qquad
    i=2,\ldots,r.
\]
The vector $e_1$ is the common type, while $e_2,\ldots,e_r$ are $k$ rare, linearly independent types.

Let $N_{\mathrm{rare}}$ denote the number of rare draws among the $m$ training samples.  Then
\[
    \mathbb{E}[N_{\mathrm{rare}}]
    =
    2\varepsilon m.
\]
If $m
    \leq
    \frac{k}{8\varepsilon}$, then
\[
    \mathbb{E}[N_{\mathrm{rare}}]
    \leq
    \frac{k}{4}.
\]
By Markov's inequality,
\[
    \Pr\left[
        N_{\mathrm{rare}}
        \geq
        \frac{k}{2}
    \right]
    \leq
    \frac{1}{2}.
\]
Hence with probability at least $1/2$,
\[
    N_{\mathrm{rare}}<\frac{k}{2}.
\]
On this event, strictly more than half of the rare basis vectors have never
appeared.  Because these vectors are linearly independent, none of the unseen
rare vectors lies in the span of the observed samples.  Therefore
\[
    R_\mu(W_m)
    >
    \frac{k}{2}
    \frac{2\varepsilon}{k}
    =
    \varepsilon.
\]
Thus
\[
    \Pr[
        R_\mu(W_m)>\varepsilon
    ]
    \geq
    \frac12.
\]
Since $\delta<1/2$, no distribution-free guarantee at confidence $1-\delta$ is possible when
$m\leq (r-1)/(8\varepsilon)$, proving
\eqref{eq:lower-rank}.

\noindent\emph{Confidence-dependent lower bound.}\par
Now consider a two-dimensional distribution
\[
    \Pr[X=e_1]=1-2\varepsilon,
    \qquad
    \Pr[X=e_2]=2\varepsilon.
\]
If the rare vector $e_2$ does not appear among the $m$ samples, then
\[
    W_m=\operatorname{span}(e_1)
\]
and therefore
\[
    R_\mu(W_m)=2\varepsilon>\varepsilon.
\]
The probability of this event is
\[
    (1-2\varepsilon)^m.
\]
For $\varepsilon\leq1/4$,
\[
    \ln(1-2\varepsilon)
    \geq
    -\frac{2\varepsilon}{1-2\varepsilon}
    \geq
    -4\varepsilon,
\]
so
\[
    (1-2\varepsilon)^m
    \geq
    e^{-4\varepsilon m}.
\]
Consequently, if
\[
    m
    <
    \frac{\ln(1/\delta)}{4\varepsilon},
\]
then
\[
    \Pr[
        R_\mu(W_m)>\varepsilon
    ]
    >
    \delta.
\]
This proves \eqref{eq:lower-confidence}.

Combining
\eqref{eq:upper-sample-complexity}, \eqref{eq:lower-rank}, and
\eqref{eq:lower-confidence}, and using
\[
    \max\{a,b\}
    \geq
    \frac{a+b}{2},
\]
gives
\[
    m^\star_{\mathrm{span}}(r,\varepsilon,\delta)
    =
    \Theta\left(
        \frac{r+\ln(1/\delta)}{\varepsilon}
    \right).
\]
\end{proof}

\subsection{Proof of the IND-CPA No-Go Theorem for Recoverable Payloads}
\label{app:proof-indcpa-nogo}

\begin{proof}
Fix a public key $pk$.  Before receiving the IND-CPA challenge, the adversary
samples
\[
    m:=2R(\lambda)-1
\]
independent values
\[
    \rho_1,\ldots,\rho_m
    \overset{\mathrm{i.i.d.}}{\leftarrow}
    \mathcal{D}_{pk}.
\]
Because the paired sampler is public, the adversary computes
\[
    (X_i,Y_i)=\mathsf{Pair}(pk;\rho_i).
\]
By \eqref{eq:pke-transport},
\[
    Y_i=A_{sk}X_i.
\]

The adversary submits any two distinct equal-length messages $m_0,m_1$.
Let the challenge ciphertext be
\[
    C^\star
    =
    \left(
        X^\star,C_{\mathrm{payload}}^\star
    \right)
    =
    \left(
        X_{\rho^\star},\,
        C_{\mathrm{payload}}^\star
    \right),
\]
where
\[
    Y^\star
    =
    Y_{\rho^\star}
    =
    A_{sk}X^\star.
\]

The adversary performs Gaussian elimination to test whether
\[
    X^\star
    \in
    \operatorname{span}(X_1,\ldots,X_m).
\]
If so, it computes coefficients $\alpha_1,\ldots,\alpha_m$ satisfying
\[
    X^\star
    =
    \sum_{i=1}^m\alpha_iX_i.
\]
Lemma~\ref{lem:coefficient-transfer} gives the exact transported value
\[
    Y^\star
    =
    \sum_{i=1}^m\alpha_iY_i.
\]
The adversary then runs the public recovery algorithm to obtain
\[
    \widetilde m
    =
    \mathsf{RecoverPayload}
    \left(
        pk,C_{\mathrm{payload}}^\star,Y^\star
    \right).
\]
If $\widetilde m=m_0$, it outputs $0$; if $\widetilde m=m_1$, it outputs
$1$; otherwise it outputs an independent uniform bit.  If $X^\star$ does not
lie in the sampled span, it also outputs an independent uniform bit.

By Proposition~\ref{prop:one-shot},
\[
\begin{split}
    p_{\mathrm{span}}
    &:=
    \Pr\left[
        X^\star
        \in
        \operatorname{span}(X_1,\ldots,X_m)
    \right]
\\
    &\geq
    1-\frac{r(pk)}{m+1}
\\
    &\geq
    1-\frac{R(\lambda)}{2R(\lambda)}
    =
    \frac12.
\end{split}
\]
Let $\mathcal{S}$ denote the sampled-span event above, and let
$\mathcal{E}$ denote the event that
\eqref{eq:recoverable-payload} correctly recovers the challenge message.
The recoverable-payload condition gives
\[
    \Pr[\neg\mathcal{E}]
    \leq
    \nu(\lambda).
\]
On $\mathcal{S}\cap\mathcal{E}$ the adversary determines $b$ with certainty,
while on $\neg\mathcal{S}$ it succeeds with probability $1/2$.  Therefore
\[
\begin{split}
    \Pr[b'=b]
    &\geq
    \Pr[\mathcal{S}\cap\mathcal{E}]
    +
    \frac12\Pr[\neg\mathcal{S}]
\\
    &\geq
    p_{\mathrm{span}}-\nu(\lambda)
    +
    \frac12(1-p_{\mathrm{span}})
\\
    &\geq
    \frac34-\nu(\lambda).
\end{split}
\]
Under the convention
\[
    \operatorname{Adv}^{\mathrm{IND\mbox{-}CPA}}
    :=
    \left|
        \Pr[b'=b]-\frac12
    \right|,
\]
the adversary therefore has advantage at least
\[
    \frac14-\nu(\lambda).
\]

Finally, the adversary uses only $O(R(\lambda))$ public paired samples and
polynomial-time linear algebra.  Since $R(\lambda)$ is polynomially bounded,
the adversary is PPT.  Since $\nu$ is negligible, the advantage remains
non-negligible, and the scheme cannot satisfy IND-CPA security.  If payload
recovery is perfectly correct, then $\nu=0$ and the advantage is at least
$1/4$.
\end{proof}

\section{Proofs for the Twisted--Skew Application}
\label{app:proofs-application}

\subsection{Proof of the Twisted--Skew Linear-Transport Lemma}
\label{app:proof-twisted-skew-transport}

\begin{proof}
Write elements of $\mathcal{R}$ as formal sums over $G$.  By the defining
product of the twisted--skew group ring
\cite{CruzMartinezMunozVillanueva2024}, homogeneous terms satisfy
\begin{equation}
    (u g)(v h)
    =
    u\,\theta_\sigma(g)(v)\,
    \alpha_\lambda(g,h)\,gh,
    \label{eq:ts-product-rule}
\end{equation}
and the product of arbitrary sums is obtained by distributive extension.
Each $\theta_\sigma(g)$ is either the identity on $\mathbb{F}_{q^2}$ or its
$q$-Frobenius automorphism.  Hence it fixes every scalar
$c\in\mathbb{F}_q$, so
\[
    \theta_\sigma(g)(cv)=c\theta_\sigma(g)(v).
\]
It follows directly from \eqref{eq:ts-product-rule} that multiplication is
additive in both inputs and $\mathbb{F}_q$-linear in each input separately.
This assertion does not require associativity.

With the parentheses kept fixed, let
\[
    L_{a_1}(x):=a_1x,
    \qquad
    R_{\widehat{\gamma_1}}(z):=z\widehat{\gamma_1}.
\]
Both maps are therefore $\mathbb{F}_q$-linear, and
$\mathcal{A}_{\mathsf{sk}}=R_{\widehat{\gamma_1}}\circ L_{a_1}$.
Explicitly, for $x,y\in\mathcal{R}$ and $c\in\mathbb{F}_q$,
\begin{align*}
    \mathcal{A}_{\mathsf{sk}}(x+y)
    &=\bigl(a_1(x+y)\bigr)\widehat{\gamma_1} \\
    &=\bigl(a_1x+a_1y\bigr)\widehat{\gamma_1} \\
    &=(a_1x)\widehat{\gamma_1}
      +(a_1y)\widehat{\gamma_1},
\end{align*}
and
\begin{align*}
    \mathcal{A}_{\mathsf{sk}}(cx)
    &=\bigl(a_1(cx)\bigr)\widehat{\gamma_1} \\
    &=\bigl(c(a_1x)\bigr)\widehat{\gamma_1} \\
    &=c\bigl((a_1x)\widehat{\gamma_1}\bigr).
\end{align*}
Thus $\mathcal{A}_{\mathsf{sk}}$ is $\mathbb{F}_q$-linear.

Finally, substituting $c_1=X_\rho$ into the verified correctness identity
\eqref{eq:ts-correctness} gives
\[
    \mathcal{A}_{\mathsf{sk}}(X_\rho)
    =
    (a_1X_\rho)\widehat{\gamma_1}
    =
    (a_2\mathsf{pk})\widehat{\gamma_2}
    =
    Y_\rho.
\]
No product has been reassociated in this argument.
\end{proof}

\subsection{Proof of the CTSP Corollary}
\label{app:proof-twisted-skew-ctsp}

\begin{proof}
Fix the first CTSP public key
\[
    pk_1=\psi((a_1,\gamma_1),h)
\]
and its associated secret linear map
$\mathcal{A}_{pk_1}(x)=(a_1x)\widehat{\gamma_1}$.  Algorithm~2 of the
original construction samples
$(a_2,\gamma_2)\leftarrow\mathsf{SK}$ independently and sets
\[
    pk_2=\psi((a_2,\gamma_2),h),
    \qquad
    k=\psi((a_2,\widehat{\gamma_2}),pk_1).
\]
Applying the correctness identity to these two key-generating pairs gives
\[
    k
    =
    (a_2pk_1)\widehat{\gamma_2}
    =
    (a_1pk_2)\widehat{\gamma_1}
    =
    \mathcal{A}_{pk_1}(pk_2).
\]

For each $i\in[m]$, the adversary independently samples
$(b_i,\eta_i)\leftarrow\mathsf{SK}$ and computes
\[
    X_i=(b_ih)\eta_i,
    \qquad
    Y_i=(b_ipk_1)\widehat{\eta_i}.
\]
The same correctness identity, with
$(a_2,\gamma_2)=(b_i,\eta_i)$, gives
\[
    Y_i
    =
    (a_1X_i)\widehat{\gamma_1}
    =
    \mathcal{A}_{pk_1}(X_i).
\]
Conditioned on the fixed $pk_1$, the vectors
$X_1,\ldots,X_m,pk_2$ are i.i.d. from $\mu_{pk_1}$.  Therefore,
Proposition~\ref{prop:one-shot} gives
\[
    \Pr\left[
        pk_2\in
        \operatorname{span}_{\mathbb{F}_q}(X_1,\ldots,X_m)
    \right]
    \geq
    1-\frac{r_{pk_1}}{m+1}.
\]
Whenever this event occurs, Gaussian elimination supplies coefficients
\[
    \alpha_1,\ldots,\alpha_m\in\mathbb{F}_q
\]
such that
\[
    pk_2=\sum_{i=1}^m\alpha_iX_i.
\]
Lemma~\ref{lem:coefficient-transfer} then recovers the CTSP target as
\[
    \widehat{k}
    :=
    \sum_{i=1}^m\alpha_iY_i
    =
    \mathcal{A}_{pk_1}(pk_2)
    =
    k.
\]
This proves \eqref{eq:ts-ctsp-bound}.

Finally, $\mathcal{R}$ has one coefficient in $\mathbb{F}_{q^2}$ for each
element of $G$.  If $|G|=2n$, then
\[
    \dim_{\mathbb{F}_q}\mathcal{R}=2\lvert G\rvert=4n,
\]
so $r_{pk_1}\leq4n$.  Taking $m=8n-1$ makes the lower bound at least $1/2$.
All samples and evaluations are public, and the remaining computation is
polynomial-time Gaussian elimination.
\end{proof}

\subsection{Proof of the Sampler-Only Break Corollary}
\label{app:proof-twisted-skew-break}

\begin{proof}
Set $pk_1=\mathsf{pk}$ and identify the first component of the challenge with
$pk_2=X^\star$.  Its mask is the corresponding CTSP target
$Y^\star=\mathcal{A}_{\mathsf{sk}}(X^\star)$.  For each independent public
encryption of zero, the adversary obtains the CTSP training pair
\[
    (X_i,c_{2,i})
    =
    (X_i,Y_i)
    =
    \bigl(X_i,\mathcal{A}_{\mathsf{sk}}(X_i)\bigr).
\]
Corollary~\ref{cor:twisted-skew-ctsp} computes an estimate
$\widehat{Y}^\star$ satisfying $\widehat{Y}^\star=Y^\star$ with probability at
least $1-r_{\mathsf{pk}}/(m+1)$.  Since
$c_2^\star=m^\star+Y^\star$, the adversary then recovers the plaintext exactly
as
\[
    m^\star=c_2^\star-\widehat{Y}^\star.
\]
This proves \eqref{eq:ts-recovery-bound}.

By definition, $\mathcal{R}$ has one coefficient in $\mathbb{F}_{q^2}$ for
each element of $G$.  If $|G|=2n$, then
\[
    \dim_{\mathbb{F}_{q^2}}\mathcal{R}=2n,
    \qquad
    \dim_{\mathbb{F}_q}\mathcal{R}=4n,
\]
and hence $r_{\mathsf{pk}}\leq4n$.  Taking $m=8n-1$ gives
\[
    1-\frac{r_{\mathsf{pk}}}{m+1}
    \geq
    1-\frac{4n}{8n}
    =
    \frac12.
\]

In the IND--CPA experiment, the adversary submits two distinct messages.  On
the sampled-span event it recovers the challenge plaintext and hence the
challenge bit; otherwise it guesses uniformly.  If the span-event probability
is $p\geq1/2$, then
\[
    \Pr[b'=b]
    \geq
    p+\frac12(1-p)
    =
    \frac12+\frac p2
    \geq
    \frac34.
\]
Its IND--CPA advantage is therefore at least $1/4$.  The attack uses only
$8n-1$ public encryptions and polynomial-time linear algebra in the
$4n$-dimensional $\mathbb{F}_q$ representation.
\end{proof}

\end{document}